\documentclass[sigplan,10pt]{acmart}

\acmConference{\hspace{-23.5cm}Pre-print}{\hspace{22.3cm}Pre-print}{March 2018}
\acmYear{\hspace{-5cm}}
\acmISBN{} 
\acmDOI{} 
\startPage{1}

\setcopyright{none}

\usepackage{booktabs}   
\usepackage{subcaption} 

\usepackage[all]{xy} 
\usepackage{marvosym}
\usepackage{amssymb,mathrsfs,stmaryrd,bbm}
\usepackage{multicol}
\usepackage{nicefrac}
\usepackage{tikz}
\usetikzlibrary{automata, positioning}

\usepackage{listings}
\newcommand{\SB}{\mathbf{SB}} 

\newcommand{\Meas}{\mathbf{Meas}}

\newcommand{\Krn}{\mathbf{Krn}}

\newcommand{\cat}[1][C]{\mathbf{C}}

\newcommand{\Pol}{\mathbf{Pol}}

\newcommand{\Klm}{{\SB_\Giry}} 
\newcommand{\BL}{\mathbf{BL}_\sigma}

\newcommand{\Lpfop}[1][p]{\mathsf{S}_{#1}} 
\newcommand{\Lpf}[1][p]{\mathsf{T}_{#1}} 
\newcommand{\Giry}{\mathsf G}

\newcommand{\abscont}[1][\cdot]{\mathcal{M}^{\ll #1}}

\newcommand{\Borel}{\mathsf{B}}
\newcommand{\ca}{\mathsf{ca}} 

\newcommand{\mr}{\mathsf{rr}}
\newcommand{\dr}{\mathsf{mr}}
\newcommand{\fr}{\mathsf{fr}}
\newcommand{\rn}{\mathsf{rn}}

\newcommand{\Lps}[1][p]{\mathrm{L}_{#1}} 
\newcommand{\one}{1}
\newcommand{\R}{\mathbb{R}}	
\newcommand{\N}{\mathbb{N}} 
\newcommand{\Ct}{2^\N} 

\newcommand{\IdMor}{\mathrm{id}} 
\newcommand{\klar}{\rightarrowtriangle} 

\newcommand{\op}{^{\mathrm{op}}}

\newcommand{\limP}{\varprojlim}
\newcommand{\klcirc}{\bullet} 
\newcommand{\Hom}[1]{\mathrm{Hom}(#1)}

\newcommand{\inv}{^{-1}}
\newcommand{\IFF}{\Leftrightarrow}
\newcommand{\Absval}[1]{\left| #1 \right|}
\newcommand{\dg}{^{\dagger}}
\newcommand{\ev}{ev}
\newcommand{\sigAlg}{\mathcal{S}}
\newcommand{\topol}{\mathcal{T}}
\newcommand{\Norm}[1]{\left\lVert #1 \right\rVert}
\newcommand{\EXP}[1]{\mathbb{E}\left[ #1 \right]}
\newcommand{\unit}{\left[0,1\right]}
\newcommand{\sod}{^\sigma} 
\newcommand{\sseq}{\hspace{-1pt}=\hspace{-1pt}} 
\newcommand{\ssint}[1]{\hspace{-1pt}\int_{#1}\hspace{-1pt}} 
\newcommand{\norm}{\mathcal{N}}
\newcommand{\SOT}{{\longrightarrow}_\mathrm{s}}

\newcommand{\lsem}{\llbracket} 
\newcommand{\rsem}{\rrbracket}

\newdir{|>}{!/4.5pt/@{|}*:(1,-.2)@^{>}*:(1,+.2)@_{>}}

\newenvironment{myProof}[1]
	{
	\noindent \textbf{Proof of #1.}\\
	}
	{
	\begin{flushright}$\blacksquare$\end{flushright}
	}

\newcommand{\pars}[1]{\left( #1 \right)} 

\theoremstyle{plain}

\theoremstyle{plain}
\newtheorem{theorem}{Theorem}
\newtheorem{proposition}[theorem]{Proposition}
\newtheorem{corollary}[theorem]{Corollary}

\newtheorem{definition}[theorem]{Definition}

\newtheorem{example}[theorem]{Example}

 \usepackage{enumitem}

\begin{document}

\title{Borel Kernels and their Approximation, Categorically}         


\author{Fredrik Dahlqvist}
\affiliation{
  \institution{University College London}            
}
\email{f.dahlqvist@ucl.ac.uk}          

\author{Vincent Danos}
\affiliation{
  \institution{Ecole Normale Sup\'{e}rieure Paris}           
}
\email{Vincent.Danos@ens.fr}         

\author{Ilias Garnier}
\affiliation{
  \institution{Ecole Normale Sup\'{e}rieure Paris}            
}
\email{Ilias.Garnier@ens.fr}          

\author{Alexandra Silva}
\affiliation{
  \institution{University College London}            
}
\email{alexandra.silva@ucl.ac.uk}          

\begin{abstract}
 This paper introduces a categorical framework to study the exact and approximate semantics of probabilistic programs. We construct a dagger symmetric monoidal category of Borel kernels where the dagger-structure is given by Bayesian inversion. We show functorial bridges between this category and categories of Banach lattices which formalize the move from kernel-based semantics to predicate transformer (backward) or state transformer (forward) semantics. These bridges are related by natural transformations, and we show in particular that the Radon-Nikodym and Riesz representation theorems - two pillars of probability theory - define natural transformations.

With the mathematical infrastructure in place, we present a generic and endogenous approach to approximating kernels on standard Borel spaces which exploits the involutive structure of our category of kernels. The approximation can be formulated in several equivalent ways by using the functorial bridges and natural transformations described above. Finally, we show that for sensible discretization schemes, every Borel kernel can be approximated by kernels on finite spaces, and that these approximations converge for a natural choice of topology.

We illustrate the theory by showing two examples of how approximation can effectively be used in practice: Bayesian inference and the Kleene $\ast$ operation of ProbNetKAT.
\end{abstract}




\maketitle
\setlength{\abovedisplayskip}{2pt}
\setlength{\belowdisplayskip}{1pt}
\setlength{\dbltextfloatsep}{9pt}
\setlength{\parskip}{2pt}

\section{Introduction}

Finding a good category in which to study probabilistic programs is a subject of active research \cite{HeunenEtAl,KozenLICS16,fossacs2017,staton2017commutative}. In this paper we present a dagger symmetric monoidal category of kernels whose dagger-structure is given by Bayesian inversion.  The advantages of this new category are two-fold.

Firstly, the most important new construct introduced by probabilistic programming, \textit{viz. Bayesian inversion}, is interpreted completely straightforwardly by the $\dagger$-operation which is native to our category. In particular we never leave the world of kernels and we therefore do not require any normalization construct. Consider for example the following simple Bayesian inference problem in Anglican (\cite{anglican2014}) 

\begin{lstlisting}
(defquery example
  (let [x (sample (normal 0 1))]
    (observe (normal x 1) 0.5)
    (> x 1)))
\end{lstlisting}

\noindent The semantics of this program is build easily and compositionally in our category:
\begin{itemize}[leftmargin=*]
\item The second line builds a Borel space equipped with a normally distributed probability measure -- an object $(\R,\mu)$ of our category.
\item The \texttt{(normal x 1)} instruction builds a Borel kernel -- a morphism $f: (\R,\mu) \to (\R,\nu)$ in our category.
\item The \texttt{observe} statement builds the Bayesian inverse of the kernel -- the morphism $f\dg: (\R,\nu) \to (\R,\mu)$ in our $\dagger$-category.
\item Finally, the kernel $f\dg$ is evaluated, i.e. the denotation of the program above is $f\dg(0.5)(\left]1,\infty\right[)$.
\end{itemize}  
\noindent The functoriality of $\dagger$ ensures compositionality.

Second, since Bayesian inference problems are in general very hard to compute (although the one given above has an analytical solution), it makes sense to seek approximate solutions, i.e. \emph{approximate denotations} to probabilistic programs. As we will show, our category of kernels comes equipped with a generic and endogenous approximating scheme which relies on its involutive structure and on the structure of standard Borel spaces. Moreover, this approximation scheme can be shown to converge for any choice of kernel for a natural choice of topology.
\paragraph{Main contributions.}  

\begin{enumerate}[leftmargin=*]
\item We build a category $\Krn$ of Borel kernels (\S \ref{sec:Krn}) and we show how two kernels which agree almost everywhere 
can be identified under a categorical quotient operation. This  technical construction is what allows us to define \emph{Bayesian inversion} as an involutive functor, denoted $\dagger$. This is a key technical improvement on \cite{fossacs2017} where the $\dagger$-structure\footnote{Suggested to us by Chris Heunen.} was hinted at but was not functorial. We show that $\Krn$ is a dagger symmetric monoidal category.  
\item  We introduce the category $\BL$ of Banach lattices and $\sigma$-order continuous positive operators as well as the K\"{o}the dual functor $(-)^\sigma: \BL\op\to\BL$ (\S\ref{sec:BL}). These will play a central role in studying convergence of our approximation schemes.

\item We provide the first\footnote{To the best of our knowledge.} categorical understanding of the Radon-Nikodym and the Riesz representation theorems. These arise as natural transformations between two functors relating kernels and Banach lattices (\S\ref{sec:BackForth}). 


\item We show how the $\dagger$-structure of $\Krn$ can be exploited to approximate kernels by averaging (\S\ref{sec:Approx}).  Due to an important structural feature of $\Krn$ (Th.~\ref{thm:ccdRepresentation}) every kernel in $\Krn$ can be approximated by kernels on {\em finite} spaces. 

\item We show a natural class of approximations schemes where the sequence of approximating kernels converges to the kernel to be approximated. The notion of convergence is given naturally by moving to $\BL$ and considering convergence in the Strong Operator Topology (\S\ref{sec:Conv}).

\item 
We apply our theory of kernel approximations to two practical applications (\S\ref{sec:Applications}). First, we show how Bayesian inference can be performed approximately by showing that the $\dagger$-operation commutes with taking approximations. 
Secondly, we consider the case of ProbNetKAT, a language developed in  \cite{2016:probnetkat, 2017:CantorScott} to probabilistically reason about networks. ProbNetKAT includes a Kleene star operator $(-)^*$ with a complex semantics which has proved hard to approximate. We show that $(-)^*$ can be approximated, and that the approximation converges.
\end{enumerate}
All the proofs can be found in the Appendix.

\paragraph{Related work.} 
Quasi-Borel sets have recently been proposed as a semantic framework for higher-order probabilistic programs in \cite{HeunenEtAl}. The main differences with our approach are: (i) unlike \cite{HeunenEtAl, staton2017commutative} we never leave the realm of kernels, and in particular we never need to worry about normalization. This makes the interpretation of \texttt{observe} statements, i.e. of Bayesian inversion, simpler and more natural. However, (ii) unlike the quasi-Borel sets of \cite{HeunenEtAl}, our category is not Cartesian closed. We can therefore not give a semantics to all higher-order programs. This shortcoming is partly mitigated by the fact that the category of Polish space, on which our category ultimately rests, does have access to many function spaces, in particular all the spaces of functions whose domain is locally compact. We can thus in principle provide a semantics to higher-order programs, provided that $\lambda$-abstraction is restricted to locally compact spaces like the reals and the integers, although this won't be investigated in this paper.

The approximation of probabilistic kernels has been a topic of investigation in theoretical computer science for nearly twenty years (see e.g.\@ \cite{desharnais2000approximating,
danos2003conditional,desharnais2004metrics,ampba}), and for much longer in the mathematical literature (e.g.\@ \cite{kim1972approximation}). Our results build on the formalism developed in \cite{ampba} with the following differences: (i) we can approximate kernels, their associated \emph{stochastic operator} (backward predicate transformer), or their associated \emph{Markov operator} (forward state transformer) with equivalent ease, and move freely across the three formalisms. (ii) Given a kernel $f: X\klar Y$, we can define its approximation $f' :X'\to Y'$ along any quotients $X'$ of $X$ and $Y'$ of $Y$ as in \cite{ampba}, but we can also `internalize' the approximation as a kernel $f^*: X\to Y$ of the original type. Morally $f'$ and $f^*$ are the same approximation, but the second approximant, being of the same type as the original kernel, can be compared with it. In particular it becomes possible to study the convergence of ever finer approximations, which we do in Section \ref{sec:Conv}. Finally, (iii) we opt to work with Banach lattices rather than the normed cones of \cite{selinger2004towards,ampba} because it allows us to formulate the operator side of the theory very naturally, and it connects to a large body of classic mathematical results (\cite{aliprantis,2012:zaanenIntroduction}) which have been used in the semantics of probabilistic programs as far back as Kozen's seminal \cite{1981:kozenProbSemantics}.

\section{A category of Borel kernels}\label{sec:Krn}


In \cite{fossacs2017} the first three authors presented a category of Borel kernels similar in spirit to the construction of this section, but with a major shortcoming. As we will shortly see, our category $\Krn$ of Borel kernels can be equipped with an involutive functor -- a \emph{dagger} operation $\dagger$ in the terminology of \cite{selinger2007dagger} -- which captures the notion of \emph{Bayesian inversion} and is absolutely crucial to everything that follows. In \cite{fossacs2017} this operation had merely been identified as a map, i.e.\@ not even as a functor. In this section we show that Bayesian inversion does indeed define a $\dagger$-structure on a more sophisticated -- but measure-theoretically very natural -- category of kernels.

\subsection{Standard Borel spaces and the Giry monad}
A \emph{standard Borel space} -- or \emph{SB space} for short -- is a measurable space $(X,\sigAlg)$ for which there exists a Polish topology $\topol$ on $X$ whose Borel sets are the elements of $\sigAlg$, i.e. such that $\sigAlg=\sigma(\topol)$ (see e.g. \cite{kechris} for an overview). Let us write $\SB$ for the category of standard Borel spaces and measurable maps. One key structural feature of $\SB$ is the following:

\begin{theorem}\label{thm:ccdRepresentation}
Every $\SB$ object is a limit of a countable co-directed diagram of finite spaces.
\end{theorem}

\noindent The \emph{Giry monad} was originally defined in two variants~\cite{giry}:
- As an endofunctor $\Giry_{\Pol}$ of $\Pol$, the category of \emph{Polish spaces}, 
one sets $\Giry_{\Pol} (X,\topol)$ to be the space of Borel probability measures over $X$ together with the weak topology. This space is Polish~\cite[Th 17.23]{kechris}, and the Portmanteau Theorem~\cite[Th 17.20]{kechris}) 
gives multiple characterizations of the weak topology.
- As an endofunctor $\Giry_{\Meas}$ of $\Meas$, the category of \emph{measurable spaces}, 
one sets $\Giry_{\Meas} (X,\sigAlg)$ to be the set of probability measures on $X$ together with the initial $\sigma$-algebra for the maps $\ev_A: \Giry_{\Meas} (X,\sigAlg)\to \R, \mu\mapsto \mu(A), A\in\sigAlg$.

In both cases the Giry monad is defined on an arrow $f: X\to Y$ as the map $f_*$ which sends a measure $\mu$ on $X$ to the \emph{pushforward measure} $f_*\mu$ on $Y$, defined as
$G(f)(\mu)(B)=f_*\mu(B):=\mu(f^{-1}(B))$ for $B$ a measurable subset of $Y$.

We want to define the Giry monad on the category $\SB$ of standard Borel spaces (and measurable maps), and the two versions of the Giry monad described above offer us natural ways to do this: given an SB space $(X,\sigma(\topol))$ we can either compute $\Giry_{\Pol}(X,\topol)$ and take the associated standard Borel space, or directly compute $\Giry_{\Meas}(X,\sigma(\topol))$. Fortunately, the two methods agree.
\begin{theorem}[\cite{kechris}, Th 17.24]
  Let $\Borel:\Pol\to\SB$ denote the functor sending a Polish space $(X,\topol)$ to its associated SB-space $(X,\sigma(\topol))$ and leaving morphisms unchanged, then
\[
  \Giry_{\Meas} \circ \Borel = \Borel \circ \Giry_{\Pol}.
\]
\end{theorem}
We define the \emph{Giry monad} on SB spaces to be the endofunctor $\Giry:\SB\to\SB$ defined by either of the two equivalent constructions above. The monadic data of $\Giry$ is given at each SB space $X$ by the unit $\delta_X: X\to\Giry X, x\mapsto \delta_x$, the Dirac $\delta$ measure at $x$, and the multiplication $m_X: \Giry^2 X\to \Giry X, \mathbb{P}\mapsto \lambda A.\int_{\Giry X} \ev_A d\mathbb{P}$. We refer the reader to \cite{giry} 
for proofs that $\delta_X$ and $m_X$ are measurable.

\subsection{The construction of $\Krn$}

Let us denote by $\Klm$ the Kleisli category associated with the Giry monad $(\Giry, \delta, m)$. We denote Kleisli arrows, i.e.\@ Markov kernels, by $X \klar Y$, and we call such an arrow \emph{deterministic} if it can be factorized as an ordinary measurable function followed by the unit $\delta$. Kleisli composition is denoted by $\klcirc$. The category $\ast \downarrow \Klm$ has arrows $\ast \klar X$ as objects, where $\ast$ is the one point SB space (the terminal object in $\SB$). An arrow from $\mu : \ast \klar X$ to $\nu : \ast \klar Y$ is a $\Klm$ arrow $f : X \klar Y$ such that $\nu = f \klcirc \mu$, i.e.\@ such that $\nu(A)=\int_X f(x)(A)d\mu$ for any measurable subset $A$ of $Y$. This situation will be denoted in short by $f : (X, \mu) \klar (Y,\nu)$, and we will call a pair $(X,\mu)$ a \emph{measured SB space}.


We want to construct a quotient of $\ast \downarrow \Klm$, such that two $\ast \downarrow \Klm$
arrows are identified if they disagree on a null set w.r.t. the measure on their domain. For $g, g' : (X,\mu) \klar (Y,\nu)$, we define $N(g,g') = \{ x \in X \mid g(x) \neq g'(x) \}$.

\begin{lemma}\label{lem:NggMeas}
  $N(g,g')$ is a measurable set.
\end{lemma}

We now define a relation $\sim$ on $\Hom{(X,\mu),(Y,\nu)}$ by saying that for any two arrows $g, g' : (X,\mu) \klar (Y,\nu)$, $g \sim g' \text{ if } \mu(N(g,g')) = 0$.
This clearly defines an equivalence relation on $\Hom{(X,\mu),(Y,\nu)}$.
In order to perform the quotient of the category $\ast \downarrow \Klm$ modulo $\sim$, we need to check that it is 
compatible with composition. 

\begin{proposition}
  \label{prop:compatibility}
  If $g \sim g'$, then $h \klcirc g \klcirc f \sim h \klcirc g' \klcirc f$.
\end{proposition}


\begin{definition}
  Let $\Krn$ be the category obtained by quotienting $\ast \downarrow \Klm$ hom-sets with $\sim$.
\end{definition}

The following Theorem is of great practical use and generalizes the well-known result for deterministic arrows.

\begin{theorem}[Change of Variables in $\Krn$]\label{thm:chgVarKrn}
Let $f: (X,\mu)\klar (Y,\nu)$ be a $\Krn$-morphism. For any 
measurable function $\phi: Y\to\R$, if $\phi$ is $\nu$-integrable, then $\phi\klcirc f(x)=\int_Y \phi\, df(x)$ is $\mu$-integrable and
\[
\int_Y \phi ~d\nu=\int_X \phi\klcirc f ~d\mu
\]
\end{theorem}

\paragraph{The symmetric monoidal structure of $\hspace{2pt}\Krn$}\hspace{-5pt}is defined on a pair of objects $(X,\mu),(Y,\nu)$ by the Cartesian product and the product of measure, i.e.\@ $(X,\mu)\otimes(Y,\nu)=(X\times Y,\mu\otimes \nu)$. On pairs of morphisms $f:(X,\mu)\klar (Y,\nu)$ and $f':(X',\mu')\klar (Y',\nu')$ it is defined by $(f\otimes f')(x,x'):= f(x)\otimes f'(x')$. The unitors, associator and braiding transformations are given by the obvious bijections.

\subsection{The dagger structure of $\Krn$}

$\Krn$ has an extremely powerful inversion principle:

\begin{theorem}[Measure Disintegration Theorem, \cite{kechris}, 17.35]\label{thm:Disintegration}
Let $f: (X,\mu)\klar (Y,\nu)$ be a \emph{deterministic} $\Krn$-morphism, there exists a unique morphism $f\dg_\mu: (Y,\nu)\klar (X,\mu)$ such that 
\begin{equation}\label{eq:DisintegrationDef}
f\klcirc f\dg_\mu=\IdMor_{(Y,\nu)}.
\end{equation}
\end{theorem}

The kernel $f\dg_\mu$ is called the \emph{disintegration of $\mu$ along $f$}. As our notation suggests, the disintegration depends fundamentally on the measure $\mu$ over the domain, however we will omit this subscript when there is no ambiguity. The following lemma relates disintegrations to conditional expectations.

\begin{lemma}[\cite{concur2016}]\label{lem:disintegrationCondExp}
Let $f: (X,\mu)\to (Y,\nu)$ be a deterministic $\Krn$-morphism, and let $\phi: X\to \R$ be measurable, then $\mu$-a.e.
\[
\phi\klcirc f\dg\klcirc f=\EXP{\phi\mid \sigma(f)}
\]
\end{lemma}

We can extend the definition of $(-)\dg$ to \emph{any} $\Krn$-morphism $f:(X,\mu)\klar (Y,\nu)$ in a functorial way, although $f\dg$ will not in general be a right inverse to $f$.
The construction of $f\dg$ is detailed in \cite{fossacs2017}, but let us briefly recall how it works. The category $\SB$ has products which are built in the same way as in $\Meas$ via the product of $\sigma$-algebras\footnote{Unlike the category $\Krn$ which does \emph{not} have products.}. Given any kernel $f:(X,\mu)\klar (Y,\nu)$, we can canonically construct a probability measure $\gamma_f$ on the product $X\times Y$  of SB-space by defining it on the rectangles of $X\times Y$ as 
\begin{equation}\label{eq:integrationDef}
\gamma_f(A\times B) = \int_{x\in X} \one_A(x) \cdot f(x)(B) ~ d\mu.
\end{equation}
Equivalently, $\gamma_f=(\delta_X\otimes f)\klcirc\Delta_X\klcirc \mu$, where $\Delta_X: X\to X\times X$ is the diagonal map.
Letting $\pi_X : X \times Y \to X$ and $\pi_Y : X \times Y \to Y$ be the canonical projections, we observe that
$\Giry \pi_X(\gamma_f) = \mu$ and $\Giry \pi_Y(\gamma_f) = \nu$: in other words,
$\gamma_f$ is a coupling of $\mu$ and $\nu$.
The disintegration of $\gamma_f$ along $\pi_Y$ is a kernel $\pi_Y\dg: (Y, \nu) \to (X\times Y, \gamma_f)$. Finally we define:
\begin{equation}\label{eq:disintegrationDef}
  f\dg=\pi_X\klcirc\pi_Y\dg.
\end{equation}
The following  $\Krn$ diagram sums up the situation:
\[
\xymatrix{
  (X,\mu)  \ar@/^1pc/@{|>}[r]_{\pi_X\dg} \ar@{|>}@<3pt>@/^2pc/[rr]_f & (X \times Y, \gamma_f) \ar@/^1pc/@{|>}[r]_{\pi_Y} \ar@/^1pc/@{|>}[l]_{\pi_X} & (Y,\nu) \ar@/^1pc/@{|>}[l]_{\pi_Y\dg}\ar@{|>}@<3pt>@/^2pc/[ll]_{f\dg}
}
\]
where $\pi\dg_X$ is explicitly given by $(\delta_X\otimes f)\klcirc\Delta_X$.
The following property characterizes the action of $(-)\dg$ on $\Krn$-morphisms:
\begin{theorem}\label{thm:bayesianinversion}
  For all $f : (X,\mu) \klar (Y,\nu)$,
  $f\dg : (Y,\nu) \klar (X,\mu)$ is the unique $\Krn$ morphism satisfying
  for all measurable sets $A \subseteq X$, $B \subseteq Y$ the following equation:  
  \begin{equation}\label{eq:disintegrationEqu}
    \int_{x\in X} \one_A(x) \cdot f(x)(B) ~ d\mu= \int_{y\in Y} f\dg(y)(A) \cdot \one_B(y) ~ d\nu
  \end{equation}
\end{theorem}

In view of Eq. (\ref{eq:disintegrationEqu}), we will call $f\dg$ the \emph{Bayesian inversion of $f$}, and refer to $(-)\dg$ as the \emph{Bayesian inversion operation} on $\Krn$. It will be crucial throughout the rest of this paper.
It is important to see that $f\dg$ absolutely depends on the choice of $\mu$ and not only on $f$ seen as a function. We can now improve on \cite{fossacs2017} and show that $(-)\dg$ is indeed a $\dagger$-operation in the strict categorical meaning of the term.

\begin{theorem}\label{thm:KrnDagger}
$\Krn$ is a dagger symmetric monoidal category, with $(-)\dg$ given by Bayesian inversion.
\end{theorem}

\section{Banach lattices}\label{sec:BL}

It is well-known that kernels can alternatively be seen as predicate -- i.e. real-valued function --transformers, or as state -- i.e. probability measure -- transformers. The latter perspective was adopted by Kozen in \cite{1981:kozenProbSemantics} to describe the denotational semantics of probabilistic programs (without conditioning). We shall see in this section and the next, that the predicate and state transformer perspectives are dual to one another in the category of \emph{Banach lattices}, a framework incidentally also used in \cite{1981:kozenProbSemantics}. For an introduction to the theory of Banach lattices we refer the reader to e.g. \cite{aliprantis,2012:zaanenIntroduction}. 

An \emph{ordered real vector space} $V$ is a real vector space together with a partial order $\le$ which is compatible with the linear structure in the sense that for all $u,v,w\in V, \lambda\in \R^+$ 
\begin{align*}
u\hspace{-1pt}\le \hspace{-1pt}v\Rightarrow u+w\hspace{-2pt}\le \hspace{-2pt} v+w \qquad\text{and}\qquad u\hspace{-1pt}\le\hspace{-1pt} v\Rightarrow \lambda u\hspace{-1pt}\le \hspace{-1pt}\lambda v
\end{align*}
An ordered vector space $(V,\le)$ is called a \emph{Riesz space} if the poset structure forms a lattice. A vector $v$ in a Riesz space $(V,\le)$ is called \emph{positive} if $0\leq v$, and its \emph{absolute value} $\Absval{v}$ is defined as $\Absval{v}=v \vee (-v)$. A Riesz space $(V,\le)$ is \emph{$\sigma$-order complete} if every non-empty countable subset of $V$ which is order bounded has a supremum.

%

A \emph{normed Riesz space} is a Riesz space $(V,\le)$ equipped with a \emph{lattice norm}, i.e.\@ a map $\Norm{\cdot}: V\to\R$ such that: 
\begin{equation}\label{eq:latticeNorm}
\Absval{v}\le\Absval{w} \text{ implies }\Norm{v}\le\Norm{w}.
\end{equation}
A normed Riesz space is called a \emph{Banach lattice} if it is (norm-) complete, i.e.\@ if  every Cauchy sequence (for the norm $\Norm{\cdot}$) has a limit in $V$. 

\begin{example}\label{ex:LpSpaces1}
For each measured space $(X,\mu)$ -- and in particular $\Krn$-objects -- and each $1\leq p\leq \infty$, the space $\Lps(X,\mu)$ is a Riesz space with the pointwise order. When it is equipped with the usual $\Lps$-norm, it is a Banach lattice. This fact is often referred to as the \emph{Riesz-Fischer theorem} (see \cite[Th 13.5]{aliprantis}). We will say that $p,q\in \N\cup\{\infty\}$ are \emph{H\"{o}lder conjugate} if either of the following conditions hold: (i) $1<p,q<\infty$ and $\frac{1}{p}+\frac{1}{q}=1$, or (ii) $p=1$ and $q=\infty$, or (iii) $p=\infty$ and $q=1$.
\end{example}


\begin{theorem}[Lemma 16.1 and Theorem 16.2 of \cite{2012:zaanenIntroduction}]\label{thm:BLsigmaComp}
Every Banach lattice is $\sigma$-order complete.
\end{theorem}

There are two very natural modes of `convergence' in a Banach lattice: \emph{order convergence} and \emph{norm convergence}. The latter is well-known, the former less so. An order bounded sequence $\{v_n\}_{n\in\N}$ in a $\sigma$-complete Riesz space (and thus in a Banach lattice) \emph{converges in order to $v$} if either of the following equivalent conditions holds:
\[
v=\liminf_n v_n:=\bigvee_n \bigwedge_{n\leq m} v_m, \hspace{1em}v=\limsup_n v_n:=\bigwedge_n \bigvee_{n\leq m} v_m.
\]
For a monotone increasing sequence $v_n$, this definition simplifies to $v=\bigvee_n v_n$, which is often written $v_n\uparrow v$. 

In a general $\sigma$-complete Riesz space, order and norm convergence are disjoint concepts, i.e.\@ neither implies the other (see \cite[Ex.~15.2]{2012:zaanenIntroduction} for two counter-examples). However if a sequence converges both in order and in norm then the limits are the same (see \cite[Th.~15.4]{2012:zaanenIntroduction}). Moreover, for \emph{monotone} sequences norm convergence implies order convergence:

\begin{proposition}[\cite{2012:zaanenIntroduction} Theorem 15.3]
If $\{v_n\}_{n\in\N}$ is an increasing sequence in a normed Riesz space and if $v_n$ converges to $v$ in norm (notation $v_n\to v)$, then $v_n\uparrow v$.
\end{proposition}

\noindent In a Banach lattice we have the following stronger property.

\begin{proposition}[Lemma 16.1 and Theorem 16.2 of \cite{2012:zaanenIntroduction}]
If $\{v_n\}_{n\in\N}$ is a sequence of positive vectors in a Banach lattice such that $\sup_n \Norm{v_n}$ converges, then $\bigvee_n v_n$ exists and $\Norm{\bigvee_n v_n}=\bigvee_n\Norm{v_n}$.
\end{proposition}

It can also happen that order convergence implies norm convergence. A lattice norm on a Riesz space is called \emph{$\sigma$-order continuous} if $v_n\downarrow 0$ ($v_n$ is a decreasing sequence whose infimum is 0) implies $\Norm{v_n}\downarrow 0$.

\begin{example}\label{ex:LpSpaces2}
For $1\leq p< \infty$, the $\Lps$-norm is \emph{$\sigma$-order continuous}, and thus order convergence and norm convergence coincide. However, for $p=\infty$ this is not the case as the following simple example shows. Consider the sequence of essentially bounded functions $v_n=\one_{[n,+\infty[}$: it is decreasing for the order on $\Lps[\infty](\R,\lambda)$ with the constant function $0$ as its infimum, i.e.\@ $v_n\downarrow 0$. However $\Norm{v_n}=1$ for all $n$. 
\end{example}

Many types of morphisms  between Banach lattices are considered in the literature but most are at least \emph{linear} and \emph{positive}, that is to say they send positive vectors to positive vectors. From now on, we will assume that all morphisms are positive (linear) operators. Other than that, we will only mention two additional properties, corresponding to the two modes of convergence which we have examined. The first notion is very well-known: a linear operator $T: V\to W$ between normed vector spaces is called \emph{norm-bounded} if there exists $C\in \R$ such that $\Norm{T v}\leq C\Norm{v}$ for every $v\in V$. The following result is familiar:
\begin{theorem}
An operator $T: V\to W$ between normed vector spaces is norm-bounded iff it is continuous.
\end{theorem}

Thus norm-bounded operators preserve norm-convergence. The corresponding order-convergence concept is defined as follows: an operator $T: V\to W$ between $\sigma$-order complete Riesz spaces is said to be \emph{$\sigma$-order continuous} if whenever $v_n\uparrow v$, $Tv=\bigvee T v_n$. It follows that we can consider two types of dual spaces on a Banach lattice $V$: on the one hand we can consider the \emph{norm-dual}:
\[
V^*=\{f: V\to \R\mid f\text{ is norm-continuous}\}
\]
and the \emph{$\sigma$-order-dual}:
\[
V\sod=\{f: V\to \R\mid f\text{ is $\sigma$-order continuous}\}
\]
The latter is sometimes known as the \emph{K\"{o}the dual} of $V$ (see \cite{1951:dieudonneKothe,2012:zaanenIntroduction}). The two types of duals coincide for a large class of Banach spaces of interest to us.
\begin{theorem}\label{thm:TwoDualsEqual}
If a Banach lattice $V$ admits a strictly positive linear functional and has a $\sigma$-order-continuous norm, then $V^*=V\sod$.
\end{theorem}

\begin{example}\label{ex:LpSpaces3}
The result above can directly be applied to our running example: given a measured space $(X,\mu)$ and an integer $1\leq p<\infty$, the Lebesgue integral provides a strictly positive functional on $\Lps(X,\mu)$, and we already know from Example \ref{ex:LpSpaces2} that $\Lps(X,\mu)$ has a $\sigma$-order-continuous norm. It follows that
\[
\Lps(X,\mu)^*=\Lps(X,\mu)\sod
\]
Moreover, it is well-known that if $(p,q)$ are H\"{o}lder conjugate and $1<p,q<\infty$, then $\Lps(X,\mu)^*=\Lps[q](X,\mu)$, and thus $\Lps(X,\mu)\sod=\Lps[q](X,\mu)$. It is also known that $\Lps[1](X,\mu)^*=\Lps[\infty](X,\mu)$, and thus $\Lps[1](X,\mu)\sod=\Lps[\infty](X,\mu)$.

However Theorem \ref{thm:TwoDualsEqual} does \emph{not} hold for $\Lps[\infty](X,\mu)$ since the $\Lps[\infty]$-norm is not $\sigma$-order continuous, as was shown in Example \ref{ex:LpSpaces2}. It is well-known that $\Lps[\infty](X,\mu)^*\neq \Lps[1](X,\mu)$, and in fact $\Lps[\infty](X,\mu)^*$ can be concretely described as the Banach lattice  $\mathrm{ba}(X,\mu)$ of \emph{charges} (i.e.\@ finitely additive finite signed measures) which are absolutely continuous w.r.t, $\mu$ on $X$ (see \cite[IV.8.16]{1971:dunfordLinear}). However, as is shown in e.g. \cite{2012:zaanenIntroduction,ampba}  
\begin{equation}
\Lps[\infty](X,\mu)\sod = \Lps[1](X,\mu)
\end{equation}
\end{example}

As Examples \ref{ex:LpSpaces2} and \ref{ex:LpSpaces3} show, the $(-)\sod$ operation brings a lot of symmetry to the relationship between $\Lps$-spaces since
\[
\Lps(X,\mu)\sod = \Lps[q](X,\mu)
\] 
for \emph{any} H\"{o}lder conjugate pair  $1\leq p\leq \infty$. For this reason we will consider the category $\BL$ whose objects are Banach lattices and whose morphisms are $\sigma$-order continuous positive operators. Note that the K\"{o}the dual of a Banach lattice is a Banach lattice, and it easily follows that $(-)^\sigma$ in fact defines a contravariant functor $\BL\op\to\BL$ which acts on morphisms by pre-composition. As we will now see, $\BL$ is the category in which predicate and state transformers are most naturally defined. 

\section{From Borel kernels to Banach lattices  }
\label{sec:BackForth}

\paragraph{The functors $\Lpfop$ and $\Lpf$.}
For $1\leq p\leq \infty$, the operation which associates to a $\Krn$-object $(X,\mu)$ the space $\Lps(X,\mu)$ can be thought of as either a contravariant or a covariant functor. We define the functors $\Lpfop: \Krn\to\BL\op, 1\leq p\leq\infty$ as expected on objects, and on $\Krn$-morphisms $f: X\klar Y$ via the well-known `predicate transformer' perspective:
\[
\Lpfop(f): \Lps(Y,\nu)\to \Lps(X,\mu), \phi\mapsto\lambda x. \int_Y \phi \,df(x)=\phi\klcirc f
\]
For a proof that this defines a functor see \cite{fossacs2017}.
We define the covariant functors $\Lpf: \Krn\to\BL, 1\leq p\leq \infty$ as $\Lpf = \Lpfop\circ (-)\dg$.


\paragraph{The functor $\abscont$.}

An \emph{ideal} of a Riesz space $V$ is a sub-vector space $U\subseteq V$ with the property that if $\Absval{u}\leq \Absval{v}$ and $v\in U$ then $u\in U$. An ideal $U$ is called a \emph{band} when for every subset $D\subseteq U$ if $\bigvee D$ exists in $V$, then it also belongs to $U$. Every band in a Banach lattice is itself a Banach lattice. Of particular importance is the band $B_v$ generated by a singleton $\{v\}$, which can be described explicitly as
\[
B_v=\{w\in V\mid (\Absval{w}\wedge n\Absval{v}) \uparrow \Absval{w}\}
\]
 
\begin{example}
Let $X$ be an SB-space and $\ca(X)$ denote the set of measures of bounded variation on $X$. It can be shown (\cite[Th 10.56]{aliprantis}) that $\ca(X)$ is a Banach lattice. The linear structure on $\ca(X)$ is as expected, the Riesz space structure is given by
\[
(\mu\vee \nu)(A)=\sup\{\mu(B)+\nu(A\setminus B)\mid B\text{ measurable }, B\subseteq A\}
\]
and the dual definition for the meet operation. The norm is given by the total variation i.e.
\[
\Norm{\mu}=\sup\left\{\sum_i^n \Absval{\mu(A_i)}\bigg| \{A_1,\ldots,A_n\}\text{ a meas. partition of }X\right\}
\]
Given $\mu\in \ca(X)$, the band $B_\mu$ generated by $\mu$ is just the set of measures of bounded variation which are absolutely continuous w.r.t. $\mu$. In particular $B_\mu$ is a Banach lattice.
\end{example}
\noindent We can now define the functor $\abscont:\Krn\to\BL$ by:
\[
\begin{cases}
\abscont(X,\mu):=B_\mu\\
\abscont f: \abscont(X,\mu)\to \abscont(Y,\nu) , \rho\mapsto f\klcirc \rho
\end{cases}
\]
We will usually write $\abscont(X,\mu)$ as $\abscont[\mu](X)$.

\begin{proposition}\label{prop:MFunctor}
  Let $f : (X,\mu) \klar (Y,\nu)$ be a $\Krn$ arrow.
  Let $\rho$ be a finite measure on $X$ such that $\rho \ll \mu$.
  Then $f \klcirc \rho \ll \nu$, and thus $\abscont$ defines a functor.
\end{proposition}

%

\paragraph{Radon-Nikodym is natural.}

We now present a first pair of natural transformations which will establish a natural isomorphism between the functors $\Lpf[1]$ and $\abscont$. First, we define the \emph{Radon-Nikodym transformation} $\rn: \abscont\to \Lpf[1]$ at each $\Krn$-object $(X,\mu)$ by the map 
\[
\rn_{(X,\mu)}: \abscont[\mu](X)\to \Lps[1](X,\mu),\quad  \rn_{(X,\mu)}(\rho)=\frac{d\rho}{d\mu}
\]
where $\nicefrac{d\rho}{d\mu}$ is of course the Radon-Nikodym derivative of $\rho$ w.r.t. $\mu$.
The fact that this transformation defines a positive operator between Banach lattices is simply a restatement of the usual Radon-Nikodym theorem \cite[III.10.7.]{1971:dunfordLinear},  combined with the well-known linearity property of the Radon-Nikodym derivative. To see that it is also $\sigma$-order-continuous, consider a monotone sequence $\mu_n\uparrow\mu$ converging in order to $\mu$ in $\abscont[\nu](X)$. This means that for any measurable set $A$ of $X$, $\lim_{n\to\infty}\mu_n(A)=\mu(A)$. Since $(\nicefrac{d\mu_n}{d\nu})_{n\in\N}$ is bounded in $\Lps[1]$-norm the function $g=\bigvee_n \nicefrac{d\mu_n}{d\nu}$ exists and is simply the pointwise limit $g(x)=\lim_{n\to\infty}\nicefrac{d\mu_n}{d\nu}(x)$. It now follows from the monotone convergence theorem (MCT) that
\[
\ssint{A} g d\nu\sseq \ssint{A} \lim_{n\to\infty} \frac{d\mu_n}{d\nu} d\nu \sseq\lim_{n\to\infty}\ssint{A}  \frac{d\mu_n}{d\nu} d\nu\sseq \lim_{n\to\infty}\mu_n(A)\sseq \mu(A)
\]
in other words, $g=\nicefrac{d\mu}{d\nu}$ and $\rn$ is well-defined. That $\rn$ is also natural has -- to our knowledge -- never been published. 

\begin{theorem}\label{thm:RNnatural}
The Radon-Nikodym transformation is natural.
\end{theorem}

Secondly, we define the \emph{Measure Representation transformation} $\dr: \Lpf[1]\to \abscont$ at each $\Krn$-object $(X,\mu)$ by the map $\dr_{(X,\mu)}:\Lpf[1] (X,\mu) \to \abscont(X,\mu)$ defined as
\[
\dr_{(X,\mu)}(f)(B_X)=\int_{B_X} fd\mu
\]
This is a very well-known construction in measure theory, and the fact that $\dr_{(X,\mu)}$ is a $\sigma$-order continuous operator between Banach lattices is immediate from the linearity of integrals and the MCT.

\begin{theorem}\label{thm:measRepNat}
The Measure Representation transformation is natural.
\end{theorem}

\paragraph{Riesz representations are natural.}

We now present a second pair of natural transformations which will establish a natural isomorphism between $(-)\sod\circ\Lpfop[\infty]$ and $\abscont$. First, we define the \emph{Riesz Representation} transformation  $\mr: (-)\sod\circ\Lpfop[\infty]\to \abscont$ at each $\Krn$-object $(X,\mu)$ by the map $\mr_{(X,\mu)}: (-)\sod\circ\Lpfop[\infty](X,\mu) \to \abscont(X,\mu)$ defined as
\[
\mr_{(X,\mu)}(F)(B_X)=F(\one_{B_X})
\]
This construction is key to a whole collection of results in functional analysis commonly known as Riesz Representation Theorems (see \cite{aliprantis} Chapter 14 for an overview). One can readily check that the Riesz Representation transformation is well-defined: $\mr_{(X,\mu)}(F)(\emptyset)=F(0)=0$ and the $\sigma$-additivity of $\mr_{(X,\mu)}(F)$ follows from the $\sigma$-order-continuity of $F$.  To see that $\mr_{(X,\mu)}(F)\ll \mu$, assume that $\mu(B_X)=0$, then clearly $\one_{B_X}=0$ $\mu$-a.e., i.e.\@ $\one_{B_X}=0$ in $\Lps[\infty](X,\mu)$, and thus $F(\one_{B_X})=0$.

\begin{theorem}\label{thm:RieszRepNatural}
The Riesz Representation transformation is natural.
\end{theorem}

Finally, we define the \emph{Functional Representation transformation}  $\fr$ at each $\Krn$-object $(X,\mu)$ by the map $\fr_{(X,\mu)}: \abscont(X,\mu)\to (-)\sod\circ\Lpfop[\infty](X,\mu)$ by
\[
\fr_{(X,\mu)}(\mu)(\phi)=\int_X \phi d\mu
\]
This construction is also completely standard in measure theory, although it has never to our knowledge been seen as a natural transformation.
\begin{theorem}\label{thm:frNatural}
The Functional Representation transformation is well-defined, i.e.\@ $\fr_{(X,\mu)}$ is a $\sigma$-order continuous positive operator, and is natural.
\end{theorem}

\paragraph{Natural Isomorphisms}We have now defined the following four natural transformations:
\[
\xymatrix
{
\Lpf[1]\ar@<1ex>@{=>}[r]^-{\dr} & \abscont\ar@<1ex>@{=>}[l]^{\rn}\ar@<1ex>@{=>}[r]^-{\fr} & (\Lpfop[\infty])\sod\ar@<1ex>@{=>}[l]^{\mr}
}
\]
In fact, both pairs form natural isomorphisms, and these can be restricted to arbitrary H\"older conjugate pairs $(p, q)$.
\begin{theorem}\label{thm:firstIso}
$\rn$ and $\dr$ are inverse of one another, in particular there exists a \emph{natural} isomorphism between $\abscont[\mu](X)$ and $\Lps[1](X,\mu)$.
\end{theorem}

\begin{theorem}\label{thm:secondIso}
$\mr$ and $\fr$ are inverse of each other, in particular there exists a \emph{natural} isomorphism between $\abscont[\mu](X)$ and $(\Lps[\infty](X,\mu))\sod$.
\end{theorem}

We can now conclude that the isomorphism proved in Theorem 6 of \cite{fossacs2017} is in fact natural.
\begin{corollary}
There exists a \emph{natural} isomorphism between $\Lpf[1]:= \Lpfop[1]\circ (-)\dg$ and $(-)\sod\circ\Lpfop[\infty]$.
\end{corollary}

We can in fact restrict this result to \emph{any} H\"{o}lder conjugate pair $(p,q)$:
\begin{theorem}\label{thm:LpLq1}
For $1\leq p\leq \infty$ with H\"{o}lder conjugate $q$, the natural transformation $\rn\circ\mr$ restricts to a natural transformation $(-)\sod\circ\Lpfop[q]\to\Lpf$.
\end{theorem}

\noindent The correspondence between the various categories and functors discussed in this section are summarized as follows:
\begin{equation}\label{diag:summary}
\xymatrix@R=11ex@C=11ex
{
 & \BL\\
\BL\op\ar[ur]^{(-)^\sigma} \ar@<-1ex>@{=>}[r]_-{\mr}  &  \ar@<-1ex>@{=>}[r]_-{\rn} \ar@<-2pt>@{=>}[l]_-{\fr}  & \Krn\op\ar[ul]_{\Lpfop[1]}\ar@<-2pt>@{=>}[l]_-{\dr}  \\
 &\Krn\ar[ur]_{(-)\dg}\ar[ul]^{\Lpfop[\infty]}\ar[uu]^>>>>>>>>{\abscont}
}
\end{equation}

\section{Approximations}\label{sec:Approx}
In this section we develop a scheme for approximating kernels
which follows naturally from the $\dagger$-structure of $\Krn$. 
Consider $f: (X,\mu)\klar (Y,\nu)$ and a pair of deterministic maps $p: (X,\mu)\to (X', p_\ast\mu)$ and $q: (Y,\nu)\to (Y',q_\ast\nu)$ (typically these maps coarsen the spaces $X$ and $Y$).
\begin{equation}\label{diag:approx}
\xymatrix@C=12ex@R=9ex
{
(X,\mu)\ar@{|>}@<5pt>[r]^f 
\ar@{|>}@<-5pt>[r]_{f^{p,q}}  
\ar@{|>}@<-1pt>@/_1pc/[d]_{p} & (Y,\nu)\ar@{|>}@<-1pt>@/_1pc/[d]_{q}\\
(X', p_\ast \mu)\ar@{|>}[r] 
\ar@{|>}@<-1pt>@/_1pc/[u]_{p\dg_\mu} 
\ar@{|>}[r]_{f_{p,q}}& (Y', q_\ast \nu)
\ar@{|>}@<-1pt>@/_1pc/[u]_{q\dg_\nu}
}
\end{equation}
The $\dagger$-structure of $\Krn$ allows us to define the new kernels
\begin{align}
&f^{p,q}:=q\dg_\nu \klcirc q \klcirc f \klcirc p\dg_\mu \klcirc p &:X\hspace{3pt}\klar Y\hspace{3pt} \label{eq:approxDef0}\\
&f_{p,q}:=q \klcirc f \klcirc p\dg_\mu &:X'\klar Y' 
 \label{eq:approxDef}
\end{align}
The supscript notation is meant to indicate that the approximation lives `upstairs' in Diagram \eqref{diag:approx} and conversely for the subscripts. Intuitively, $f_{p,q}$ and $f^{p,q}$ take the average of $f$ over the fibres given by $p,q$ according to $\mu$ and $\nu$ (see Section \ref{sec:Applications} for concrete calculations). The advantage of \eqref{eq:approxDef} is that we can approximate a kernel on a huge space by a kernel on a, say, finite one. The advantage of \eqref{eq:approxDef0} is that although it is more complicated, it is morally equivalent and has the same type as $f$, which means that we can compare it to $f$.

A very simple consequence of our definition is that Bayesian inversion commutes with approximations. We shall use this in \S \ref{sec:Applications:Bayes} to perform \emph{approximate Bayesian inference}.
\begin{theorem}\label{thm:ApproxBayesInverse}
Let $f:(X,\mu)\klar (Y,\nu)$, let $p: X\to X'$ and $q: Y\to Y'n$ be a pair of deterministic maps, then
\[
(f\dg)^{q,p}=(f^{p,q})\dg\text{ and }(f\dg)_{q,p}=(f_{p,q})\dg
\]
\end{theorem}

In practice we will often consider endo-kernels $f: X\klar  X$ with a single coarsening map $p: X\to X'$ to a \emph{finite space}. In this case \eqref{eq:approxDef0} simplifies greatly.

\begin{proposition}\label{prop:approxEndo}
Under the situation described above
\begin{equation}\label{eq:approxEndo}
f^p:= p\dg_\nu \klcirc p \klcirc f \klcirc p\dg_\mu \klcirc p = f \klcirc p\dg_\mu \klcirc p
\end{equation}
\end{proposition}
 In the case covered by Proposition \ref{prop:approxEndo}, the interpretation of $f^p$ is very natural: for each $x\in X$ the measure $f(x)$ is approximated by its average over the fibre to which $x$ belongs, conditioned on being in the fibre. For fibres with strictly positive $\mu$-probability, this is simply
 \[
 f^p(x)(A)=\frac{\int_{y\in p\inv(p(x))} f(y)(A)~d\mu}{\mu(p\inv(p(x))}
 \]
However \eqref{eq:approxEndo} also covers the case of $\mu$-null fibres. 
Note also that in the case where $f^p = f$, the map $p$ corresponds to
what is known as a strong functional bisimulation for $f$.

\paragraph{Approximating is non-expansive.} It is well-known that conditional expectations are non-expansive and we know from Lemma \ref{lem:disintegrationCondExp} that pre-composing by $p\dg_\mu\klcirc p$ as in \eqref{eq:approxEndo} amounts to conditioning. The following lemma is an easy consequence.

\begin{lemma}\label{lem:nonExpansive}
Let $f: (X,\mu)\klar(Y,\nu)$ and $q: X\to X'$ be a deterministic quotient, then for all $1\leq p\leq \infty$ and $\phi\in\Lps(Y,\nu)$
\[
\Norm{\Lpfop f^q(\phi)}_p\leq \Norm{\Lpfop f(\phi)}_p
\]
\end{lemma}

\paragraph{Compositionality of approximations} 
In the case where we wish to approximate a composite kernel $g \klcirc f$, it
might be convenient, for modularity reasons, to approximate $f$ and $g$ separately.
This does not entail any loss of information provided the quotient maps are
hemi-bisimulations, in the following sense. Let $p : X \to X', q : Y \to Y', r : Z \to Z'$ 
be deterministic quotients and let $f: (X,\mu)\klar(Y,\nu), g: (Y,\nu) \klar (Z,\rho)$
be composable kernels. We say that $q$ is a left hemi-bisimulation for $f$ if
$f = q\dg \klcirc q \klcirc f$, and conversely that it is a right hemi-bisimulation
for $g$ if $g = g \klcirc q\dg \klcirc q$ holds. In either case, one
can verify using Theorems \ref{thm:Disintegration} and~\ref{thm:ApproxBayesInverse}
that approximation commutes with composition, i.e.\@ that
$(g \klcirc f)^{p,r} = g^{q,r} \klcirc f^{p,q}$.

\paragraph{Discretization schemes} We will use \eqref{eq:approxDef} and \eqref{eq:approxEndo} to build sequences of arbitrarily good approximations of kernels. For this we introduce the following terminology.
\begin{definition}\label{def:discretisation}
We define a \emph{discretization scheme for an SB-space} $X$ to be a countable co-directed diagram (ccd) of finite spaces for which $X$ is a cone (not necessarily a limit).
\end{definition}
\noindent If $(X_i)_{i\in I}$ is a discretization scheme of $X$ and $p_i: X\to X_i$ are the maps making $X$ a cone, then it follows from the definition that if $i<j$, $\sigma(p_i) \subseteq \sigma(p_{j})$
where $\sigma(p_i)$ is the $\sigma$-algebra generated by $p_i$. For each $i\in I$ the finite quotient $p_i$ defines a \emph{measurable partition} of $X$ whose disjoint components $p_i\inv(\{k\}),k\in X_i$ we will call \emph{cells}. 

By Theorem \ref{thm:ccdRepresentation} every SB-space has a discretization scheme for which it is not just a cone but a limit.

In practice we will work with discretization schemes linearly ordered by $\N$. In this case the sequence $(X,\sigma(p_n))_{n\in \N}$ defines what probabilists call a \emph{filtration} and we will denote the approximation  $f^{p_n}$ given by \eqref{eq:approxEndo} simply by $f^n$.

\section{Convergence}\label{sec:Conv}
We now turn to the question of convergence of approximations. There appears to be little literature on the subject of the convergence of approximations of Markov kernels. One rare reference is \cite{kim1972approximation}. 
Via the functor $\Lpfop$ defined above in Sections \ref{sec:BL} and \ref{sec:BackForth} 
we can seek a topology in terms of the \emph{operators associated to a sequence of kernels}.
Indeed, following~\cite{kim1972approximation}, we will prove convergence results for the \emph{Strong Operator Topology} (SOT).

\begin{definition}

We will say that a sequence of kernels $f^n: X\klar Y$ \emph{converges to $f: X\klar Y$ in strong operator topology}, and write $f^n\SOT f$, if \hspace{1pt} $\Lpfop[1] f^n$ converges to $\Lpfop[1] f$ in the strong operator topology, i.e.\@ if 
\[
\lim_{n\to\infty} \Norm{\Lpfop[1] f^n(\phi)- \Lpfop[1] f(\phi)}_1=0
\]
\end{definition}


\paragraph{Proving convergence.} We start with the following key lemma which is a consequence of  L\'{e}vy's upward convergence Theorem (\cite[Th. 14.2]{williams1991probability}) . 

\begin{lemma}\label{lem:Levy}
Let $f:(X,\mu)\klar (Y,\nu)$ be a $\Krn$-morphism and let $p_n: X\to X_n, n\in\N$ be a discretization scheme such that for $\mathcal{B}_X$ the Borel $\sigma$-algebra of $X$ we have
\[
\mathcal{B}_X=\sigma\left(\bigcup_n \sigma(p_n)\right)
\]
and let $A\subseteq Y$ be measurable, then for $f^n:=f\klcirc p_n\dg\klcirc p_n$
\[
\lim_{n\to\infty} f^n(x)(A)=f(x)(A)
\]
for $\mu$-almost every $x\in X$. Moreover, 
\[
\lim_{n\to\infty} \Norm{\Lpfop[1]f^n (\one_A)-\Lpfop[1]f(\one_A)}_1=0
\]
\end{lemma}

\begin{theorem}[Convergence of Approximations Theorem]\label{thm:ApproxConv}
Under the conditions of Lemma \ref{lem:Levy}, for $\mu$-almost every $x\in X$
\[
\lim_{n\to\infty} f^n(x)(A)=f(x)(A)
\]
for all Borel subsets $A$. Moreover, 
\[
\lim_{n\to\infty} \Norm{\Lpfop[1]f^n (\phi)-\Lpfop[1]f(\phi)}_1=0
\]
for any $\phi\in \Lps[1](X,\nu)$. In other words $f^n\SOT f$.
\end{theorem}

%

Note that 
operators of the shape $\Lpfop f^n$ obtained from a discretization scheme are \emph{finite rank operators}. Thus, we, in fact, also obtained a theorem to approximate stochastic operators by \emph{stochastic operators of finite rank} for the SOT topology. In general, we cannot hope for convergence in the stronger norm topology since the identity operator -- which is stochastic -- is a limit of operators of finite rank in the norm topology iff the space is finite dimensional. 

Note also 
that the various relationships established in Section \ref{sec:BackForth} allow us to move from an approximation of a kernel to an approximation of the corresponding Markov operator. Since a discretization scheme making $f^n\SOT f$ will also make $(f\dg)^n \SOT f\dg$, it follows from Theorem \ref{thm:firstIso} that we get a finite rank approximation of the Markov operator $\abscont(f)$. 

\section{Applications}\label{sec:Applications}

\subsection{Approximate Bayesian Inference}\label{sec:Applications:Bayes}
Consider again the inference problem from the introduction.
%
%
There one needed to invert $f(x)=\norm(x,1)$ with prior $\mu=\norm(0,1)$.
We can use Theorem \ref{thm:ApproxBayesInverse} to see how our approximate Bayesian inverse compares to the exact solution which in this simple case is known to be $f^\dagger_\mu(0.5)=\norm(1/4,1/2)$. To do this, 
we use a doubly indexed discretization scheme:
\[q_{mn}: \R\to 2 \times m \times n + 2\]
defining a window of width $2m$ 
centred at $0$ 
divided 
in $2 m n$ equal intervals; with the remaining intervals
$(-\infty,-m]$ and $(m,\infty)$ each sent to a point (hence the $+2$ above). 

\begin{figure}[h!]
  \includegraphics[width=0.4\textwidth]{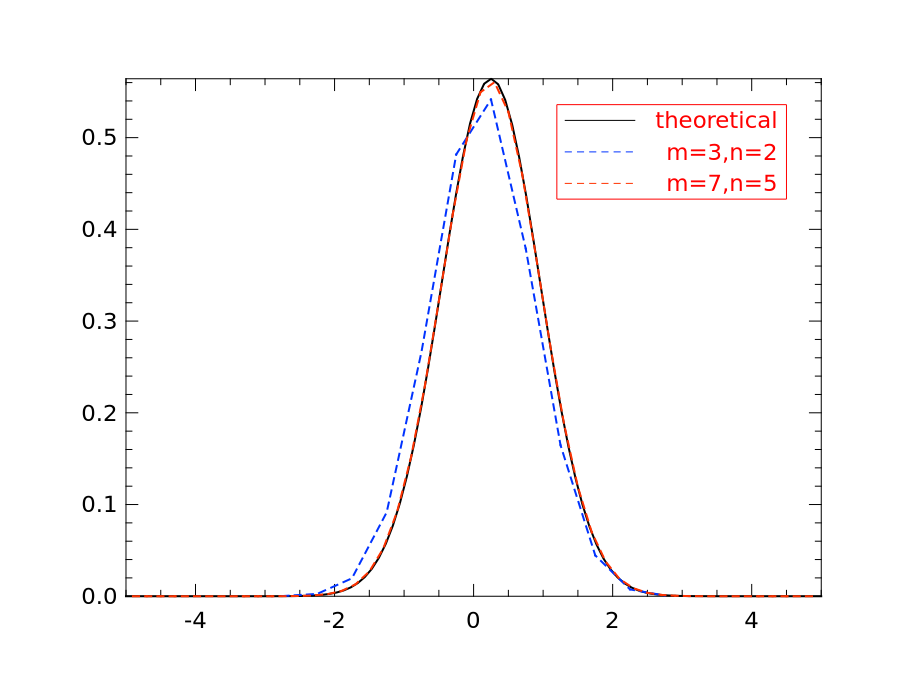}
  \caption{Approximate posteriors}
  \label{fig:approximate_posteriors}
\end{figure}

Since all classes induced by $q_{mn}$ have positive $\mu$-mass,
approximants can be computed simply as:
\[
f^{m,n}([k])([l]) = {\mu[k]}^{-1}\int_{x \in [k]} \norm(x,1)([l])~d\mu
\]
where $[k]$, $[l]$ range over classes of $q_{mn}$.
The corresponding stochastic matrices are shown
in Fig.~\ref{fig:f53} and~\ref{fig:f610} for $m,n=5,3$ and $6,10$ respectively.

\begin{figure*}
  \begin{minipage}[b]{.46\linewidth}
    \centering\includegraphics[width=0.7\textwidth]{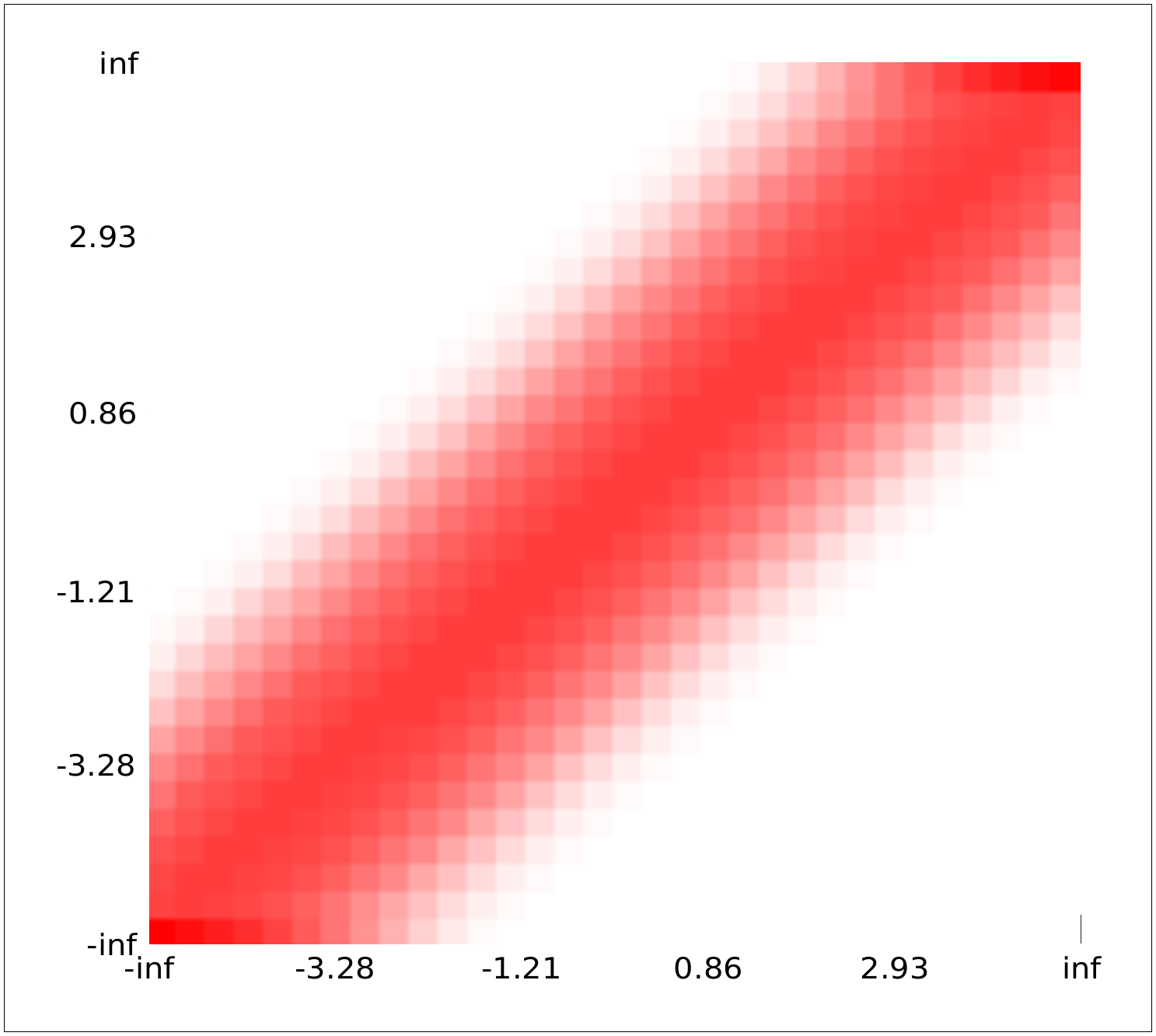}
    \caption{Log-likelihood of $f^{5,3}$}
    \label{fig:f53}
 \end{minipage} \hfill
  \begin{minipage}[b]{.46\linewidth}
    \centering\includegraphics[width=0.7\textwidth]{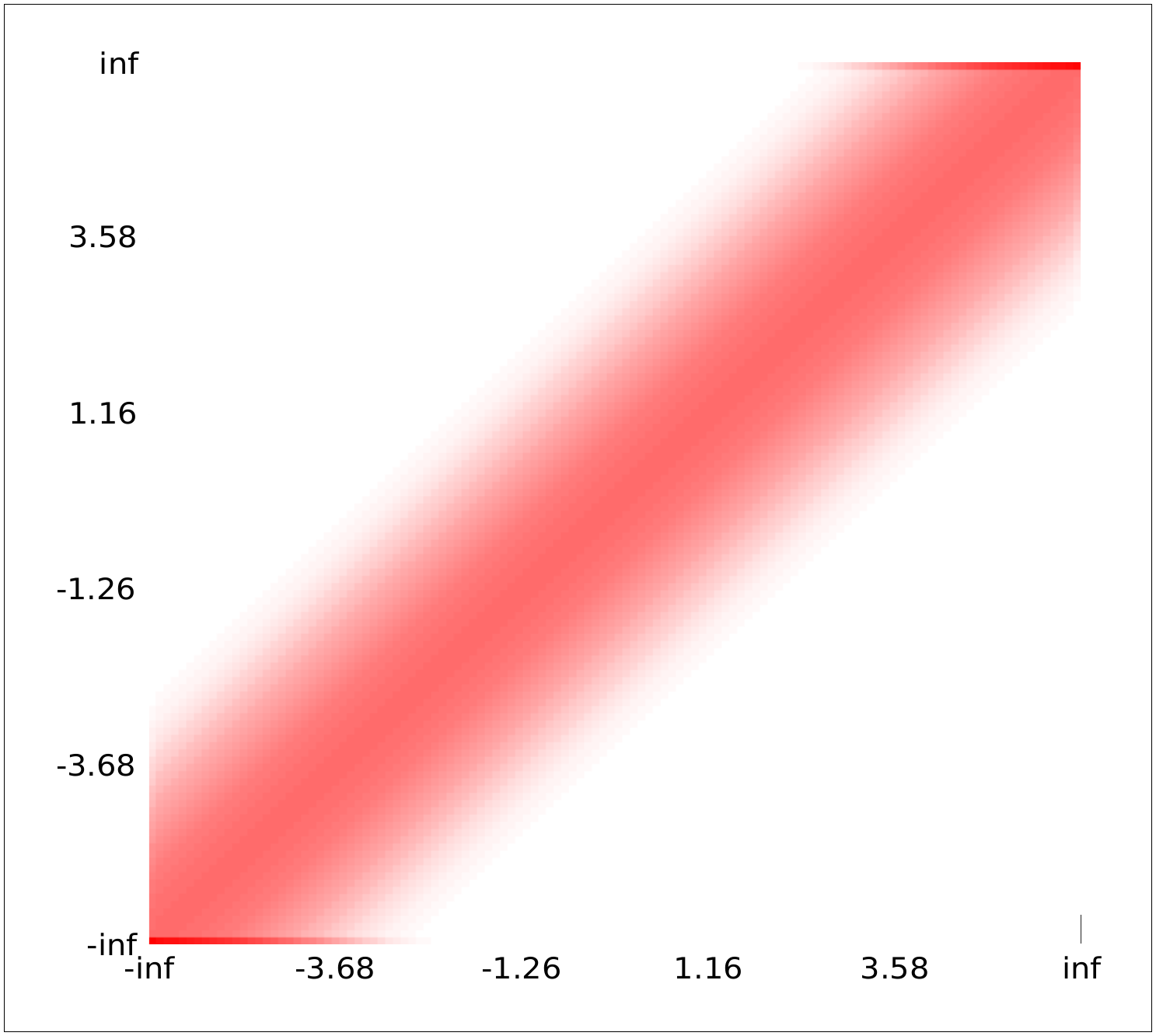}
    \caption{Log-likelihood of $f^{6,10}$}
    \label{fig:f610}
 \end{minipage} \hfill
\end{figure*}

Since these approximants are finite, 
their Bayesian inverse can be computed 
directly by Bayes theorem (i.e.\@ taking the adjoint of
the stochastic matrices):
\begin{equation}
\label{eq:bayesianinverse}
{f^{m,n}}\dg([l])([k]) = \frac{\mu[k] \cdot f^{m,n}([k])([l])}{\nu[l]}
\end{equation}
with $\nu=f_*(\mu)$. Commutation of inversion and approximation guarantees that 
the $f^{m,n}{}\dg$ converge to $f\dg$. 

Indeed, Fig.~\ref{fig:approximate_posteriors} shows the
the Lebesgue density of $f^{m,n}{}\dg(0.5)$ for $m,n = 3,2$ (in dashed blue) and $7,5$ 
(dashed red). The latter approximant is already hardly distinguishable from 
the exact solution (solid black).

It must be emphasized that this example is meant only as an illustration and does not
constitute a universal solution to the irreducibly hard (not even computable in general
~\cite{AckermanFreerRoyLICS11}) problem of performing Bayesian inversion.
Also, not all quotients are equally convenient: what makes the approach computationally
tractable is that the fibres are easily described and the measure conveniently evaluated
on such fibres. 

\subsection{Approximating the Kleene star of ProbNetKAT}\label{sec:Applications:ProbNetKAT}
\newcommand{\cantor}{\mathtt{cantor}}
ProbNetKAT (\cite{2016:probnetkat,2017:CantorScott}) is a probabilistic network specification language extending Kleene Algebras with Tests (\cite{kozen1997kleene}) with network primitives and a binary probabilistic choice operator $\oplus_\lambda, \lambda\in\unit$. For the purpose of the example shown here we will not need to introduce the full syntax and semantics of ProbNetKAT, rather we will focus on a single ProbNetKAT program which we will call $\cantor$ and is given by:
\begin{equation}\label{eq:probnetkat}
\cantor:=\tt p;(dup;p)^*\quad\text{where}\quad p:=\pi_0! \oplus_{\nicefrac{1}{2}} \pi_1!
\end{equation}
The program acts on sets of finite sequences of 0 and 1, which can be thought of as packet histories. 
We will write $H$ for the set $\{0,1\}^*$ of all packet histories and $H_n$ for the set of histories of length as most $n$. A ProbNetKAT program is always interpreted as a kernel $2^H\to \Giry 2^H$. Programs with both $\tt dup$ and $^\ast$ revealed to be quite complex from the earliest development of the language. As we will describe, $ \cantor$ denotes a continuous distribution and hence having a way to approximate it is crucial for practical uses of the language. The denotation of $\pi_0!$ on a single sequence $\{\mathtt{(a_0,\ldots,a_n)}\}$ is:
\[
 \tt \lsem\pi_0!\rsem(\{(a_0,\ldots,a_n)\})=\delta_{\{(0,a_1,\ldots,a_n)\}}
\]
in other words $\pi_0!$ overwrites the first entry in the sequence with $\tt 0$. Similarly, $\pi_1!$ overwrites the first entry with $\tt 1$. This semantics is extended to sets of sequences in the obvious way by taking direct images.  The semantics of $\mathtt{p}$ is thus:
\[
\lsem \mathtt{p}\rsem(a)=\text{0.5}\delta_{\lsem\pi_0!\rsem(a)}+\text{0.5}\delta_{\lsem\pi_1!\rsem(a)}
\]
The denotation of $\tt dup$ is given on singleton histories by
\[
\tt \lsem dup \rsem(\{(a_0,\ldots,a_n)\})=\delta_{\{(a_0,a_0,\ldots,a_n)\}}
\]
i.e. $\tt dup$ shifts the history to the right and duplicates the first entry. Again, this is extended to sets of histories by taking direct images. The sequential composition operator $;$ is interpreted by Kleisli composition. 

The interpretation of the Kleene star is more involved, and we here describe it categorically. To avoid any confusion we will not use Kleisli arrows in this construction, i.e.\@ all kernels will be explicitly typed as kernels. Note first that the infinite product $(2^H)^\omega$ can be defined as the limit of the ccd given by the maps $q_{n+1,n}:(2^H)^{n+1}\to (2^H)^n$ dropping the last component. By Bochner's theorem  (\cite{machine})  this also holds of $\Giry((2^H)^\omega)$. Next, consider any program $\tt r$. We turn $2^H$ into a cone for the diagram with limit $\Giry((2^H)^\omega)$ via the inductively defined maps:
\begin{align}
a_1&=\eta\otimes \lsem \mathtt{r}\rsem\klcirc \Delta_1: 2^H\to  \Giry(2^H\times 2^H)\label{eq:inftyDef1}\\
a_n&= a_{n-1}\otimes \lsem \mathtt{r}\rsem\klcirc \Delta_n: (2^H)^n\to  \Giry((2^H)^{n}\times 2^H)\label{eq:inftyDef2}
\end{align}
where $\Delta_n: (2^H)^n\to (2^H)^{n}\times 2^H$ is the map copying the last entry.
It is easy to check $q_{n+1,n}\klcirc a_n\klcirc a_{n-1}\klcirc\ldots\klcirc a_1=a_{n-1}\klcirc\ldots\klcirc a_1$, and the diagram described by the morphisms $ b_n:=a_{n}\klcirc\ldots\klcirc a_1: 2^H\to \Giry (2^H)^n$ makes $2^H$ a cone for $\limP\Giry (2^H)^n$. There must therefore exist a unique morphism
\[
\lsem  \mathtt{r}\rsem_\infty: 2^H\to \Giry\left((2^H)^\infty\right).
\]
For each input, this kernel builds a distribution on the sample paths of the discrete-time stochastic processes associated with $\tt r$ and this input. We now define
\[
\lsem \mathtt{r}^*\rsem :=\Giry\left( \bigcup\right)\circ \lsem \tt r\rsem_\infty
\]
where $\bigcup: (2^H)^\infty\to 2^H$ is the map taking infinitary unions. Since the definition above makes sense for any kernel $f$ on $2^H$, we will overload the Kleene star and put $f^*:=\Giry\bigcup \circ f_\infty$. Given the input $\tt(0)$, a sample path of $\cantor$ will draw uniformly a history of size 1, then a history of size 2 whose suffix matches the size 1 history drawn at the previous step, and so on for every integer. The distribution $\lsem \cantor\rsem\mathtt{(0)}$ associates to a measurable collection of sets of histories $A$ the probability that the union of a sample path from $\tt(0)$ belongs to $A$. For example $\lsem \cantor\rsem \mathtt{(0)}\{A\mid (\mathtt{01})\in A\}=\nicefrac{1}{4}$, since there's a $\nicefrac{1}{4}$ chance that a sample path will have drawn  $(\mathtt{01})$ amongst the histories of size 2. 


We start by turning $2^H$ into a $\Krn$-object.
Consider the countable directed diagram given by all injections $i_{mn}: H_m\to H_n, n>m$, then $H=\varinjlim H_n$, and it follows that $2^H=\limP 2^{H_n}$ since $2^-$ turns colimits into limits. We know from Bochner's theorem that $\Giry 2^H=\limP \Giry 2^{H_n}$, and we use this fact to place a canonical measure on $2^H$ as follows: since each $2^{H_n}$ is finite with cardinality $c_n:=2^{\sum_i^n 2^i}=2^{2^{n+1}-1}$, and can thus be equipped with the uniform measure $\nicefrac{1}{c_n}$, we can find a limit measure $\mu$ on $2^H$ with the pleasing property that for all history truncating maps $p_n: 2^H\to 2^{H_n}$, the pushforward $\mu_n:=(p_n)_*\mu$ is the uniform measure on $H_n$. It is clear that these maps define a discretization scheme on $2^H$ which satisfies the condition of Theorem \ref{thm:ApproxConv}. We will now show that if $f^n\SOT f$, then $(f^n)^\ast\SOT f^\ast$. To prove this we need the following lemma which is interesting in its own right.

\begin{lemma}\label{lem:tensorCont}
The monoidal structure of $\Krn$ is continuous for the SOT, i.e.\@ $f^n\SOT f$ and $g^n\SOT g$ implies $f^n\otimes g^n\SOT f\otimes g$.
\end{lemma}

\begin{theorem}\label{thm:KleeneStarConv}
Under the set-up described above, for any kernel $f: (2^H,\mu)\klar (2^H,\nu)$ we have $
(f^{n})^\ast \SOT f^\ast$
\end{theorem}
The advantage of working over finite spaces is that $(f^n)^\ast$ can, in principle at least, be computed for kernels defined in ProbNetKAT. Let us examine this in the case of $\cantor$ and of the discretization scheme $p_n: 2^H\to 2^{H_n}$.

In the case $n=3$ the underlying Markov chain has $2^7$ states, but has an interesting property which means we need not consider them all: when we compute $\tt \lsem p\rsem  ; \lsem (dup ; p)\rsem_\infty^3$, the process necessarily lands in an ergodic component of the chain consisting of the singletons of histories of length exactly 3. The reason is that once the process reaches histories of length 3 it starts randomly re-writing the histories, and with probability 1 any two histories will eventually get re-written to the same thing. Once a set of histories has decreased in cardinality by one, it can never go back, thus eventually \emph{any} set of histories gets re-written to a single length 3 history, and then loops among length 3 singletons indefinitely.
The situation is represented from the initial state $\tt(0)$ in Figure \ref{fig:MC} where, for clarity's sake, the ergodic component is symbolized by common double-sided arrows to a new state.

\begin{figure}[h]
\begin{tikzpicture}
font=\scriptsize,
    baseline=1ex,shorten >=.1pt,node distance=6mm,on grid,
    semithick,auto,
every state/.style={fill=white,draw=black,circular drop shadow,inner sep=0mm,text=black}
        \node[state]             (s) {$\tt(0)$};
        
        \node[state, above right = 15mm and 6mm of s] (0) {$\tt(0)$};
        \node[state, below right = 15mm and 6mm of s] (1) {$\tt(1)$};
        
         \node[state, above right = 5mm and 10mm of 0] (00) {$\tt(00)$};
         \node[state, below right = 5mm and 10mm of 0] (10) {$\tt(10)$};
         \node[state, above right = 5mm and 10mm of 1] (01) {$\tt(01)$};
         \node[state, below right = 5mm and 10mm of 1] (11) {$\tt(11)$};
         
         \node[state, above right = 0mm and 20mm of  00] (000) {$\tt(000)$};
         \node[state, below right  = 0mm and 20mm of  00] (100) {$\tt(100)$};
         \node[state, above right = 0mm and 20mm of 10] (010) {$\tt(010)$};
         \node[state, below right = 0mm and 20mm of 10] (110) {$\tt(110)$};
         \node[state, above right = 0mm and 20mm of 01] (001) {$\tt(001)$};
         \node[state, below right = 0mm and 20mm of 01] (101) {$\tt(101)$};
         \node[state, above right = 0mm and 20mm of 11] (011) {$\tt(011)$};
         \node[state, below right = 0mm and 20mm of 11] (111) {$\tt(111)$};
         
          \node[state, below right = 0mm and 20mm of 110] (t) {};

        \draw[every loop]
            (s) edge[bend left, auto=left]  node {0.5} (0)
            (s) edge[bend right, auto=right]  node {0.5} (1)
            (0) edge[bend left, auto=left]  node {0.5} (00)
            (0) edge[bend right, auto=right]  node {0.5} (10)
            (1) edge[bend left, auto=left]  node {0.5} (01)
            (1) edge[bend right, auto=right]  node {0.5} (11)
            (00) edge[bend left=20, auto=left]  node {0.5} (000)
            (00) edge[bend right=20, auto=right]  node {0.5} (100)
            (10) edge[bend left=20, auto=left]  node {0.5} (010)
            (10) edge[bend right=20, auto=right]  node {0.5} (110)
            (01) edge[bend left=20, auto=left]  node {0.5} (001)
            (01) edge[bend right=20, auto=right]  node {0.5} (101)
            (11) edge[bend left=20, auto=left]  node {0.5} (011)
            (11) edge[bend right=20, auto=right]  node {0.5} (111)
            (000) edge[bend left, auto=left,<->]  node {} (t)
            (001) edge[bend right, auto=left,<->]  node {} (t)
            (010) edge[bend left, auto=left,<->]  node {} (t)
            (011) edge[bend right, auto=left,<->]  node {} (t)
            (100) edge[bend left, auto=left,<->]  node {} (t)
            (101) edge[bend right, auto=left,<->]  node {} (t)
            (110) edge[bend left, auto=left,<->]  node {} (t)
            (111) edge[bend right, auto=left,<->]  node {} (t);
\end{tikzpicture}
\caption{ $\tt \lsem p\rsem  ; \lsem (dup ; p)\rsem_\infty^3)$ starting at $\tt(0)$}\label{fig:MC}
\end{figure}
\noindent By post-composing with $\Giry\bigcup$ we have
\[
 \lsem \mathtt{p}\rsem  ; (\lsem \mathtt{(dup ; p)}\rsem^3)^\ast(a)(A)=0
 \] 
if $A$ does not contain all histories of length 3. We have meaningful answers to questions about histories up to length 2:
\begin{align*}
&\lsem \mathtt{p}\rsem  ; (\lsem \mathtt{(dup ; p)}\rsem^3)^\ast(\mathtt{(0)})(\{A\mid \mathtt{(1)}\in A\})=0.5\\
&\lsem \mathtt{p}\rsem  ; (\lsem \mathtt{(dup ; p)}\rsem^3)^\ast(\mathtt{(0)})(\{A\mid \mathtt{(10)}\in A\})=0.25
\end{align*}
In other words, at $n=3$ we have the first two steps in the construction of the Cantor distribution towards which $\cantor$ converges.
\vspace{-2ex}
\section{Conclusion}
We have presented a framework for the exact \emph{and} approximate
semantics of first-order probabilistic programming. 
The semantics can be read off 
either in terms of kernels
between measured spaces, or in terms of operators between $\Lps$
spaces. Either forms come with related involutive structures:
Bayesian inversion for (measured) kernels between Standard Borel
spaces, and K\"othe duality for positive linear and $\sigma$-continuous
operators between Banach lattices. Functorial relations between both
forms can themselves be related by way of natural isomorphisms.
Our main result is the
convergence of general systems of finite approximants in terms of the
strong operator topology (the SOT theorem). Thus, in principle, one
can compute arbitrarily good approximations of the semantics of a
probabilistic program of interest for any given (measurable) query. 
%
%
Future work may allow one to derive stronger notions of convergences
given additional Lipschitz control on kernels, or to
develop approximation schemes that are adapted to
the measured kernel of interest. More
ambitiously perhaps, one could investigate whether MCMC sampling schemes
commonly used to perform approximate Bayesian inference in the
context of probabilistic programming could be seen as randomized
approximations of the type considered in this paper.


\vspace{-2ex}

\bibliography{RN}

\appendix
\section{Appendix}

\begin{myProof}{Theorem \ref{thm:ccdRepresentation}}
This is a consequence of the Isomorphism Theorem (Theorem 15.6 of \cite{kechris}): two SB spaces are isomorphic iff they have the same cardinality. Uncountable SB spaces are thus all isomorphic to the Cantor space $\Ct$ which is the limit of the countable co-directed diagram $(2^n)_{n\in\N}$ with the connecting morphisms $p_{n+1,n}: 2^{n+1}\to 2^n$ truncating binary words of length $n+1$ at length $n$. Similarly all SB-spaces of cardinality $\aleph_0$ are isomorphic to the one-point compactification of $\N$, which is the limit of the countable co-directed diagram $(n)_{n\in\N}$ with the connecting morphisms $p_{n+1,n}: n+1\to n, i\mapsto \min(i,n)$. The case of finite SB spaces is trivial.
\end{myProof}

\begin{myProof}{Lemma \ref{lem:NggMeas}}
  By Dinkyn's $\pi$-$\lambda$ theorem, two finite measures are equal if and only if they agree
  on a $\pi$-system generating the $\sigma$-algebra. Any standard Borel space admits such a countable $\pi$-system (any countable basis for a Polish topology generating the $\sigma$-algebra). Let $\{B_n\}_{n \in \N}$ be such a $\pi$-system. Then, for all $x \in X$, $g(x) \neq g'(x) \Leftrightarrow \exists n. g(x)(B_n) \neq g'(x)(B_n)$. Hence,
\[
  \begin{array}{lll}
    N(g,g') & = &  \cup_n \{x \in X \mid g(x)(B_n) \neq g'(x)(B_n)\} \\
            & = & \cup_n \{ x \in X \mid ev_{B_n}(g(x)) \neq ev_{B_n}(g'(x)) \} \\
            & = & \cup_n (ev_{B_n} \circ g - ev_{B_n} \circ g')\inv(\R \setminus \{0\}) \\
  \end{array}
\]
  By definition of the measurable structure of $\Giry(Y)$, $ev_{B_n} \circ g - ev_{B_n} \circ g'$ is
  measurable, hence $N(g,g')$ is also measurable.
\end{myProof}

\begin{myProof}{Proposition \ref{prop:compatibility}}
	We first show that if $g \sim g'$, then $h \klcirc g \sim h \klcirc g'$. Clearly, for   any space $V$ and any deterministic function $u : Y \to V$,
  $N(u \circ g, u \circ g') \subseteq N(g, g')$. 
  By definition of the Kleisli category, $h \klcirc g = m_Z \circ \Giry(h) \circ g$
  and similarly for $h \klcirc g'$. Taking $u = m_Z \circ \Giry(h)$,
  we obtain that $\mu(N(h \klcirc g, h \klcirc g')) \le \mu(N(g, g'))$.

  It is now enough to show that $\lambda(N(g \klcirc f, g' \klcirc f)) = 0$.
  Let us reason contrapositively. We have:
   \[
  \begin{array}{llll}
          & \lambda(N(g \klcirc f, g' \klcirc f))                        & > 0 \\
     \IFF & \int_{w \in W} \one_{N(g \klcirc f, g' \klcirc f)}(w) \: d\lambda & > 0 \\
     \to & \int_{w \in W} \sum_{n \in \N} \one_{(g \klcirc f)(w)(B_n) \neq (g' \klcirc f)(w)(B_n)} \: d\lambda & > 0 \\
     \stackrel{\exists n}{\to} & \int_{w \in W} \one_{(g \klcirc f)(w)(B_n) \neq (g' \klcirc f)(w)(B_n)} \: d\lambda & > 0 \\
     \to & \int_{w \in W} \Absval{(g \klcirc f)(w)(B_n) - (g' \klcirc f)(w)(B_n)} \: d\lambda & > 0   \\
     \IFF & \int_{w \in W} \int_{x \in X} \Absval{g(x)(B_n) - g'(x)(B_n)} ~df(w) \: d\lambda & > 0   \\
     \IFF & \int_{x \in X} \Absval{g(x)(B_n) - g'(x)(B_n)} \: d\mu & > 0  \\ 
     \stackrel{\exists X^+ \subseteq X}{\hspace{-13pt}\to} & \int_{x \in X^+} g(x)(B_n) - g'(x)(B_n) \: d\mu & > 0  \\
     \to & \int_{x \in X^+} \one_{g(-)(B_n) > g'(-)(B_n)}(x) \: d\mu & > 0  \\
     \to & \int_{x \in X^+} \one_{N(g,g')}(x) \: d\mu & > 0  \\    
  \end{array}
  \]
  The last line implies $\mu(N(g,g')) > 0$, a contradiction.
\end{myProof}

\begin{myProof}{Theorem \ref{thm:chgVarKrn}}
If $\phi$ is $\nu$-integrable, there exists a monotone sequence $\{\phi_n\}$ of simple functions such that $\phi_n\uparrow \phi$ and $\int_Y \phi_n d\nu\to \int_Y \phi d\nu<\infty$. By definition each $\phi_n=\sum_{i=0}^k\alpha_i \one_{B_i}$, and by unravelling the definition we have
\begin{align*}
\int_Y \one_{B_i} d\nu& =\nu(B_i)\\
&=\int_X f(x)(B_i)d\mu\\
&=\int_X \int_Y \one_{B_i}df(x) d\mu\\
&=\int_X (\one_{B_i}\klcirc f) d\mu
\end{align*}
From which it follows that
\[
\int_Y \phi_nd\nu=\int_X \int_Y \sum_{i=0}^k\alpha_i\one_{B_i}df(x)d\mu=\int_X(\phi_n\klcirc f)d\mu  
\]
and the result follows from the Monotone Convergence Theorem (MCT).
\end{myProof}

\begin{myProof}{Theorem \ref{thm:bayesianinversion}}
  It follows by definition of $f\dg$ and from the disintegration theorem that
  \begin{equation}
    \label{eq:thm:bayesianinversion}
    \int_{y\in Y} f\dg(y)(A) \cdot \one_B(y) ~ d\nu = \gamma_f(A \times B),
  \end{equation}
  from which Eq.~\ref{eq:disintegrationEqu} follows easily.
  It remains to prove that this uniquely characterizes $f\dg$. Let us reason contrapositively.
  Assume there exists $g : (Y,\nu) \klar (X,\mu)$
  verifying for all $A,B$ measurable $\int_{y \in Y} g(y)(A) \cdot \one_B(y) = \gamma_f(A \times B)$ as in Eq.~\ref{eq:thm:bayesianinversion}
  and such that $\nu(N(f\dg,g)) > 0$ (assuming we take some representative of $f\dg$). Let $\{ A_n \}_{n \in \N}$
  be a countable $\pi$-system generating the $\sigma$-algebra of $X$. It is enough to test equality of
  measures on $X$ on this $\pi$-system. Therefore, $N(f\dg,g) = \cup_n \{y \mid f\dg(y)(A_n) \neq g(y)(A_n) \}$.
  Since $\nu(N(f\dg,g)) > 0$, there must exist a $k \in \N$ such that $\nu(\{y \mid f\dg(y)(A_k) \neq g(y)(A_k) \}) > 0$.
  Therefore, $N_k^+ =  \{y \mid f\dg(y)(A_k) > g(y)(A_k) \}$ must also have positive measure for $\nu$.
  But then, $\int_{y \in Y} g(y)(A_k) \cdot \one_{N_k^+}(y) \neq \gamma_f(A_k \times N_k^+)$, a contradiction.
\end{myProof}

\begin{myProof}{Theorem \ref{thm:KrnDagger}}
  Let us first show that $(-)\dg$ is a functor $\Krn\to\Krn\op$, i.e.\@ that $\IdMor_{(X,\mu)}\dg=\IdMor_{(X,\mu)}$ and that for any $f: (X,\mu)\klar (Y,\nu)$ and $g:(Y,\nu)\klar(Z,\rho)$ we have $(g\klcirc f)\dg=f\dg\circ g\dg$.

  Let $(X, \mu)$ be an object of $\Krn$ and $\IdMor_{X,\mu}$ the corresponding
  identity. By Th.~\ref{thm:bayesianinversion}, it is enough to prove, for all $A, A'$ measurable
  subsets of $X$, that
  \[
  \int_{x \in X} \one_{A}(x) \cdot id_{X,\mu}(x)(A') ~d\mu = \int_{x \in X} id_{X,\mu}(x)(A) \cdot \one_{A'}(x) ~d\mu.
  \]
  We have:
  \[
    \int_{x \in X} \one_{A}(x) \cdot id_{X,\mu}(x)(A') ~d\mu = \int_{x \in X} \one_{A}(x) \cdot \one_{A'} ~d\mu = \mu(A \cap A')
  \]
  The same calculation on the right hand side of the first equation yields trivially the same result.
  Hence the equality is verified.

  Now, on to compatibility w.r.t. composition.
  In sight of Th.~\ref{thm:bayesianinversion}, it is enough to show that for all
  $A \subseteq X$, $C \subseteq Z$,
  \[
  \int_{x \in X} (g \klcirc f)(x)(C) \cdot \one_{A}(x) ~d\mu = \int_{z \in Z} \one_{C}(z) \cdot (f\dg \klcirc g\dg)(z)(A) ~d\rho
  \]
  In the following, for $X$ a measurable space, we denote by $SF(X)$ the set of simple functions over $X$
  (finite linear combinations of indicator functions of measurable sets). We will use repeatedly the monotone
  convergence theorem (MCT). The left hand side of the above equation can be re-written as:
  \[
  \begin{array}{ll}
    & \int_{x \in X} \pars{\int_{y \in Y} g(y)(C) ~df(x)} \cdot \one_{A}(x) ~d\mu \\
    \stackrel{(1)}{=} & \int_{x \in X} \pars{\int_{y \in Y} \lim_{n \to \infty}  g_n(y) ~df(x)} \cdot \one_{A}(x) ~d\mu  \\
    \stackrel{(2)}{=} & \int_{x \in X} \lim_{n \to \infty} \pars{\int_{y \in Y} g_n(y) ~df(x)} \cdot \one_{A}(x) ~d\mu \\
  \end{array}
  \]
  where $(1)$ is because $g_n \uparrow g(-)(C), g_n \in SF(Y)$ and $(2)$ by monotone convergence.
  Note that the $n$-indexed family $x \mapsto \int_{y \in Y} g_n(y) ~df(x)$ is pointwise increasing. Therefore,
  \[
  \begin{array}{ll}
   ^{(\ast)}& \lim_{n \to \infty}  \int_{x \in X} \pars{\int_{y \in Y} g_n(y) ~df(x)} \cdot \one_{A}(x) ~d\mu  \\
    \stackrel{(1)}{=} &  \lim_{n \to \infty} \int_{x \in X} \pars{\sum_{i=1}^{k_n} \alpha_i^n f(x)(C_i^n)} \cdot \one_{A}(x) ~d\mu \\
    = &  \lim_{n \to \infty} \sum_{i=1}^{k_n} \alpha_i^n \int_{x \in X}  f(x)(C_i^n) \cdot \one_{A}(x) ~d\mu \\
     \stackrel{(2)}{=} &  \lim_{n \to \infty} \sum_{i=1}^{k_n} \alpha_i^n \int_{y \in Y}  \one_{C_i^n}(y) \cdot f\dg(y)(A) ~d\nu  \\
    \stackrel{(\ast)}{=} &  \int_{y \in Y}  g(y)(C) \cdot f\dg(y)(A) ~d\nu \\
     \stackrel{(3)}{=} &  \int_{y \in Y}  g(y)(C) \cdot \lim_n f_n(y) ~d\nu \\
    \stackrel{(\ast)}{=} & \lim_n \sum_{i = 1}^{k_n} \beta_i^n \int_{y \in Y}  g(y)(C) \cdot \one_{D_i^n}(y)  ~d\nu \\ 
    \stackrel{(2)}{=} & \lim_n \sum_{i = 1}^{k_n} \beta_i^n \int_{z \in Z} \one_{C}(z) \cdot g\dg(z)(D_i^n) ~d\rho  \\
    = & \lim_n \int_{z \in Z} \one_{C}(z) \cdot \int_{y \in Y} \sum_{i = 1}^{k_n} \beta_i^n \one_{D_i^n}(y) ~dg\dg(z) ~d\rho \\
    \stackrel{(\ast)}{=} & \int_{z \in Z} \one_{C}(z) \cdot \int_{y \in Y} f\dg(y)(C) ~dg\dg(z) ~d\rho  \\
    = & \int_{z \in Z} \one_{C}(z) \cdot (f\dg \klcirc g\dg)(z)(A) ~d\rho
  \end{array}
  \]
  where $(\ast)$ is by monotone convergence, $(1)$ is because $g_n \in SF(Y)$, $(2)$ is by Th.~\ref{thm:bayesianinversion} and $(3)$ is because $f_n \uparrow f\dg(-)(A), f_n \in SF(Y)$.  We have proved the sought identity.

  Finally let us show that $(-)\dg$ is involutive, i.e.\@ that for any $f: (X,\mu)\klar (Y,\nu)$, $(f\dg)\dg=f$. 
  This follows easily by two applications of Th.~\ref{thm:bayesianinversion}):
  we have
  \[
  \begin{array}{lll}
    \int_{x \in X} \one_{A}(x) \cdot (f\dg)\dg(x)(B) ~d\mu & = &
    \int_{y \in Y} f\dg(y)(A) \cdot \one_{B}(y) ~d\nu \\
 & = &
    \int_{x \in X} \one_{A}(x) \cdot f(x)(B) ~d\mu;
  \end{array}
  \]
  and since adjoints are unique, $f = (f\dg)\dg$.
  
  The fact that $(f\otimes g)\dg=f\dg\otimes g\dg$ follows immediately from the definitions and the property of disintegrations given by Th.~\ref{thm:bayesianinversion}. The fact that the associator, unitors and braiding transformations are unitary follows immediately from the fact that they are deterministic isomorphisms and Th.~\ref{thm:Disintegration}.
\end{myProof}

\begin{myProof}{Proposition \ref{prop:MFunctor}}
  Let $B \subseteq Y$ be a measurable set. By definition,
  we have $(f \klcirc \rho)(B) = \int_X ev_B \circ f ~d\rho$
  where we recall that $ev_B : \Giry(X) \to \R_+$ is the evaluation
  morphism. Let $\{ f^B_n \}_{n \in \N}$ be an increasing chain of
  simple functions converging pointwise to $ev_B \circ f$ such that
  for each $n$,  $f^B_n = \sum_{i = 1}^{k_n} \alpha_i^n \one_{A_i^n}$
  with $\alpha_i^n \ge 0$.
  By the MCT,
  \[
  (f \klcirc \rho)(B) = \lim_n \int_X f^B_n ~d\rho = \lim_n \sum_{i = 1}^{k_n} \alpha_i^n \rho(A_i^n).
  \]
  Similarly,
  \[
  \nu(B) = (f \klcirc \mu)(B) = \lim_n \int_X f^B_n ~d\mu = \lim_n \sum_{i = 1}^{k_n} \alpha_i^n \mu(A_i^n).
  \]
  Notice that since the integral is linear and the sequence
  $\{ f_n^B \}_n$ is increasing, the sequences
  $\{\int_X f^B_n ~d\rho\}_n$  and $\{\int_X f^B_n ~d\mu\}_n$
  are also increasing.
  Assume $\nu(B) = 0$. Then for all $n$, $\int_X f^B_n ~d\mu = 0$.
  We deduce that for all $n$, for all $1 \le i \le k_n$, either
  $\alpha_i^n = 0$ or $\mu(A_i^n) = 0$.
  Using that $\rho \ll \mu$, we deduce that
  for all $1 \le i \le k_n$, either $\alpha_i^n = 0$ or $\rho(A_i^n) = 0$,
  from which we conclude that for all $n$, $\int_X f^B_n ~d\rho = 0$
  and finally, $(f \klcirc \rho)(B) = 0$.
  Hence, $f \klcirc \rho \ll \nu$.
\end{myProof}

\begin{myProof}{Theorem \ref{thm:RNnatural}}
We start by proving the following Lemma
\begin{lemma}\label{lem:rnHelp}
For any $f:(X,\mu)\klar(Y,\nu)$, $\phi\in \Lps[1](Y,\nu)$, and $B_X\subseteq X$ measurable 
\[
\int_{B_X} \left(\int_Y \phi df(x)\right)~d\mu = \int_Y \phi(y) f\dg(y)(B_X)  ~d\nu
\]
\end{lemma}
\begin{proof}
We start by showing the equation on characteristic functions. If $B_Y$ is measurable in $Y$, we have
\begin{align*}
\int_{B_X} \left(\int_Y \one_{B_Y} df(x)\right)~d\mu & =\int_{B_X} f(x)(B_Y)~d\mu \\
& = \int_Y \one_{B_Y}(y) f\dg(y)(B_X)~d\nu & \text{Eq. (\ref{eq:disintegrationEqu})}
\end{align*}
Since $\phi$ is measurable and integrable, there exists a sequence $\phi_n \uparrow \phi$ of simple functions such that $\lim_n \int_Y\phi_n~d\nu<\infty$, and the results follows by the linearity of integration and the MCT. 
\end{proof}

\noindent We can now prove the naturality of $\rn$. Let $f:(X,\mu)\klar(Y,\nu)$ be a morphism in $\Krn$; we have on the one hand
\begin{align*}
\rn_{(Y,\mu)}\circ \abscont(f)(\rho)(y)&=\rn_{(Y,\nu)}(\int_X f(x)(-)~d\rho)(y)\\
&=\frac{d \int_X f(x)(-)d\rho}{d\nu}(y) & (*)
\end{align*}
and on the other
\begin{align*}
\Lpf[1](f\dg)\circ \rn_{(X,\mu)}(\rho)(y)&=\Lpf(f\dg)\left(\frac{d\rho}{d\mu}\right)(y)\\
&=\int_X \frac{d\rho}{d\mu}~df\dg(y) & (**)
\end{align*}
To show the equality of these two maps in $\Lps[1](Y,\nu)$ it is enough to show that they are equal $\nu$-a.e. To see this, we show that $(**)$ satisfies the condition to be the Radon-Nikodym derivative $(*)$. Let $B_Y$ be a measurable subset of $Y$. We have from the well-known property of Radon-Nikodym derivatives:
\[
\int_{B_Y} \frac{d \int_X f(x)(-)~d\rho}{d\nu} d\nu=\int_{x\in X} f(x)(B_Y)~d\rho
\]
Moreover, we have
\begin{align*}
\int_{B_Y}\int_X \frac{d\rho}{d\mu}~df\dg(y)~d\nu & \stackrel{(1)}{=}\int_{x\in X} \frac{d\rho}{d\mu}(x) f(x)(B_Y)~d\mu \\
&\stackrel{(2)}{=}\int_{x\in X} f(x)(B)~d\rho
\end{align*}
where $(1)$ is by Lemma \ref{lem:rnHelp} and $(2)$ is a well-known property of Radon-Nikodym derivatives.
\end{myProof}

\begin{myProof}{Theorem \ref{thm:measRepNat}}
We start with the following elementary lemma.
\begin{lemma}\label{lem:drHelp}
If $\psi,\phi\in \Lps[1](X,\mu)$ then
\[
\int_X \psi\phi~d\mu=\int_X \psi~d(\dr_{(M,\mu)}\phi)
\]
\end{lemma}
\begin{proof}
\noindent The proof of naturality now follows easily: it is enough to show the equality in the case where $\psi=\one_{B_X}$ for a measurable subset $B_X$ of $X$, and the result then extends to all measurable functions by linearity of integrals and the MCT. We have
\begin{align*}
\int_X \one_{B_X}\phi~d\mu=\int_{B_X}\phi d\mu&:= \dr_{(M,\mu)}(\phi)(B_X)\\
&=\int_X \one_{B_X}~d(\dr_{(M,\mu)}(\phi))
\end{align*}
\end{proof}

\noindent To show naturality we now let $f:(X,\mu)\klar (Y,\nu)$ be a $\Krn$-morphism, $\phi\in \Lps[1](X,\mu)$ and $B_Y$ measurable in $Y$
\begin{align*}
&\dr_{(Y,\nu)} \Lpf[1](f\dg)(\phi)(B_X)\\
&=\dr_{(Y,\nu)}( \phi\klcirc f\dg)(B_Y)\\
&=\int_{B_Y} \phi\klcirc f\dg~d\nu\\
&=\int_{B_Y} \int_X \phi df\dg(y)~d\nu\\
&= \int_X  f(x)(B_Y) \phi(x)~d\mu &\text{Lemma \ref{lem:rnHelp}} \\
&=\int_X f(x)(B_Y)~d(\d(\dr_{(X,\mu)}(\phi))&\text{Lemma \ref{lem:drHelp}}\\
&=\abscont f\circ  \dr_{(X,\mu)}(\phi)(B_Y)
\end{align*}
\end{myProof}

\begin{myProof}{Theorem \ref{thm:RieszRepNatural}}
Again, we start with a simple but helpful Lemma.
\begin{lemma}\label{lem:mrHelp}
Let $F\in (\Lpfop[\infty](X,\mu)\sod$ and $\phi\in \Lpfop[\infty](X,\mu)$, then
\[
F(\phi)=\int_X \phi~d(\mr_{(X,\mu)}(F))
\]
\end{lemma}
\begin{proof}
Starting with characteristic functions, let $\phi=\one_B$ for some measurable subset $B$ of $X$. We then have
\[
F(\one_B):=\mr_{(X,\mu)}(F)(B)=\int_X \one_B~d(\mr_{(X,\mu)}(F))
\]
We can then extend the result to simple functions by linearity and then to all functions in $\Lps[\infty](X,\mu)$ by the MCT.
\end{proof}

\noindent To show naturality we now let $f:(X,\mu)\klar  (Y,\nu)$ be a $\Krn$-morphism, $F\in (\Lpfop[\infty](X,\mu))\sod$ and $B_Y$ measurable in $Y$. We have
\begin{align*}
\abscont f\circ \mr_{(X,\mu)}(F)(B_Y)&= \int_X f(x)(B_Y) ~d(\mr_{(X,\mu)}(F))\\
&=F(f(\cdot)(B_Y)) \qquad \text{Lemma \ref{lem:mrHelp}}\\
&=F(\int_X \one_{B_Y} df(\cdot))\\
&=F(\one_{B_Y}\klcirc f)\\
&=\mr_{(Y,\nu)}(F(-\klcirc f))(B_Y)\\
&=\mr_{(Y,\nu)} \circ (\Lpf[1]f)\sod (F)(B_Y)
\end{align*}
\end{myProof}

\begin{myProof}{Theorem \ref{thm:frNatural}}
We start by showing that $\fr$ is well defined.
The linearity of $\fr_{(X,\mu)}$ is easily checked on simple functions and extended by the CMT. Positivity is also immediate. For the $\sigma$-order continuity, let $\mu_m\uparrow\mu$, $\phi\in \Lps[\infty](X,\mu)$, and $\phi_n\uparrow \phi$ be a monotone approximation of $\phi$ by simple functions. We need to show that 
\[
\lim_{m\to\infty}\int_X \phi ~d\mu_m =\int_X \phi ~d\mu
\]
For note first that the doubly indexed series $\int_X \phi_n d\mu_m$ is monotonically increasing in $m$, since the $\mu_m$ are monotonically increasing. Note also that the differences 
\[
d_{mn}:=\int_X \phi_n~d\mu_{m+1}-\int_X \phi_n~d\mu_{m}
\]
are monotonically increasing in $n$. Indeed we have
\begin{align*}
&\left(\ssint{X} \phi_{n+1}~d\mu_{m+1}-\ssint{X} \phi_{n+1}~d\mu_{m}\right)-\left(\ssint{X} \phi_{n}d\mu_{m+1}-\ssint{X} \phi_{n}~d\mu_{m}\right)\\
&=\int_X (\phi_{n+1}-\phi_n)~d\mu_{m+1}-\int_X(\phi_{n+1}-\phi_n)~d\mu_{m}>0
\end{align*}
since the sequences $\phi_n$ and $\mu_m$ are monotonically increasing. Since $d_{mn}$ is monotonically increasing in $n$ we can apply the CMT to $d_{mn}$ seen as a function of $m$ w.r.t. the counting measure, i.e.\@ 
\[
\lim_{n\to\infty} \sum_{m}^\infty~d_{mn}=\sum_{m}^\infty \lim_{n\to\infty}~d_{mn}
\]
which is to say, by taking partial sums
\begin{align*}
\lim_{n\to\infty} \lim_{m\to\infty}\sum_{k=1}^m d_{kn}&=\lim_{n\to\infty}\lim_{m\to\infty}\int_X\phi_n~d\mu_m \\
&=\lim_{m\to\infty} \sum_{k=1}^m \lim_{n\to\infty}d_{kn}\\
&=\lim_{m\to\infty} \lim_{n\to\infty}\int_X\phi_n~d\mu_m
\end{align*}
which concludes the proof that $\fr$ is well-defined.

We now prove naturality.
Let $f:(X,\mu)\klar (Y,\nu)$ be a $\Krn$-morphism, $\rho\in \abscont[\mu](X)$ and $\phi\in \Lpfop[\infty](Y,\nu)$ we then have
\begin{align*}
\fr_{(Y,\nu)}\circ \abscont f(\rho)(\phi)&=\int_Y \phi~d(\abscont f(\rho)(\phi))\\
&=\int_Y (\phi\klcirc f)~d\rho \qquad \text{Theorem }\ref{thm:chgVarKrn}\\
&= \Lpfop[\infty](f) \left(\int_Y (-)~d\rho\right)(\phi)\\
&= \Lpfop[\infty] (f) \circ \fr_{(X,\mu)}(\rho)(\phi)
\end{align*}
\end{myProof}

\begin{myProof}{Theorem \ref{thm:firstIso}}
The fact that $\rn_{(X,\mu)}$ and $\dr_{(X,\mu)}$ are inverse of each other is just a restatement of the two well-known equalities for Radon-Nikodym derivatives:
\[
\frac{d\int_- \phi ~d\mu}{d\mu}=\phi\hspace{1ex}\text{ and }\hspace{1ex}\int_{B_X} \frac{d\rho}{d\mu}~d\mu=\rho(B_X)
\]
\end{myProof}

\begin{myProof}{Theorem \ref{thm:secondIso}}
Let $(X,\mu)$ be a $\Krn$-object, let $F\in (\Lps[\infty](X,\mu))\sod$ and let $\phi\in \Lps[\infty](X,\mu)$. We have
\[
\fr_{(X,\mu)}\circ \mr_{(X,\mu)}(F)(\phi)=\int_X \phi~d(\mr_{(X,\mu)}(F))=F(\phi)
\]
where the last equality follows from Lemma \ref{lem:mrHelp}.
Similarly, we have
\[
\mr_{(X,\mu)}\circ \fr_{(X,\mu)}(\rho)(B_X)=\fr_{(X,\mu)}(\rho)(\one_{B_X)}=\int_X \one_{B_X}d\rho=\rho(B_X)
\]
\end{myProof}

\begin{myProof}{Theorem \ref{thm:LpLq1}}
The case $p=1$ has been treated already, for the case of $1<p<\infty$, see for example the proof of Theorem 4.4.1 of \cite{bogachev1}. Finally for the case of $p=\infty$, see Proposition 3.3 of \cite{ampba}. 
\end{myProof}

\begin{myProof}{Proposition \ref{prop:approxEndo}}
Note in \eqref{eq:approxEndo} that we disintegrate $p$ with respect to two different measures. For notational clarity let us define the endo-kernels
\begin{align*}
a&:= p\dg_\mu \klcirc p\\
b& :=p\dg_\nu \klcirc p
\end{align*}
The kernel $a$ associates to each $x\in X$ in a fibre $p\inv(\{i\})$ the measure $\pi_\mu\dg(i)$ supported by this fibre. In particular it is constant on each fibre, and similarly for $b$. We can now compute:
\begin{align*}
&f^p(x)(A)\\
&:=b\klcirc f\klcirc a(x)(A)\\
&=\int_{y\in X}b(y)(A)~d(f\klcirc a)(x)\\
&\stackrel{(1)}{=}\sum_{i\in X'} \int_{y\in p\inv(i)}b(y)(A)~d(f\klcirc a)(x) \\
&\stackrel{(2)}{=}\sum_{i\in X'} p_\nu\dg(i)(A\cap p\inv(i)) f\klcirc a(x)(p\inv(i))\\
&=\sum_{i\in X'} p_\nu\dg(i)(A\cap p\inv(i)) \int_{y\in X} f(y)(p\inv(i)) ~da(x)\\
&\stackrel{(3)}{=}\sum_{i\in X'} p_\nu\dg(i)(A\cap p\inv(i)) \int_{y\in p\inv(p(x))} f(y)(p\inv(i)) ~dp_\mu\dg(p(x))
\end{align*}
where $(1)$ follows by decomposing $X$ in fibres, $(2)$ is because $b(y)(A)$ is constant on fibres, and $(3)$ uses the fact that $a(x):=p_\mu\dg(p(x))$ is supported on the fibre of $p(x)$. We can considerably simplify the expression above. Note first that by definition of the disintegration
\begin{align}
\nu(A\cap p\inv(i))=p\dg_\nu(i)(A\cap p\inv(i))\nu(p\inv(i)) \label{eq:simplifying1}
\end{align}
Similarly, by definition of the disintegration, for $i,j\in X'$
\begin{align}
&\int_{p\inv(j)}f(x)(p\inv(i))~d\mu\nonumber  \\
& = \int_{X} 1_{p\inv(j)}(x)f(x)(p\inv(i))~d\mu\nonumber \\
&= \sum_{k\in X'} \mu(p\inv(k))\int_{p\inv(k)}1_{p\inv(j)}(x)f(x)(p\inv(i))~dp\dg_\mu(k) \nonumber  \\
&=\mu(p\inv(j))\int_{p\inv(j)}f(x)(p\inv(i))~dp\dg_\mu(j)  \label{eq:simplifying2}
\end{align}
where the last step uses the fact that $p\dg_\mu(k)$ is supported by the fibre over $k$.
By multiplying the LHS of \eqref{eq:simplifying1}, \eqref{eq:simplifying2} we get
\begin{align}
& \nu(A\cap p\inv(i))\int_{p\inv(j)}f(x)(p\inv(i))d\mu \nonumber \\
& = \int_X f(x)(A\cap p\inv(i)) d\mu \int_{p\inv(j)}f(x)(p\inv(i))~d\mu \nonumber \\
&= \int_{p\inv(j)} \hspace{-1em}f(x)(A\cap p\inv(i))~d\mu \int_X f(x)(p\inv(i))~d\mu \nonumber \\
&= \mu(p\inv(j))\int_{p\inv(j)} \hspace{-1em}f(x)(A\cap p\inv(i))~dp\dg_\mu(j) \ssint{X} f(x)(p\inv(i))~d\mu\nonumber\\
&=\mu(p\inv(j))\nu(p\inv(A))\int_{p\inv(j)} f(x)(A\cap p\inv(i))~dp\dg_\mu(j) 
 \label{eq:simplifying3}
\end{align}
It now follows from \eqref{eq:simplifying1}, \eqref{eq:simplifying2}, and \eqref{eq:simplifying3} that
\begin{align*}
&\mu(p\inv(j))\nu(p\inv(A))\int_{p\inv(j)}\hspace{-1em} f(x)(A\cap p\inv(i)) ~dp\dg_\mu(j)  \\
&=p\dg_\nu(i)(A\cap p\inv(i))\nu(p\inv(i)) \mu(p\inv(j))\int_{p\inv(j)}\hspace{-1em}f(x)(p\inv(i))~dp\dg_\mu(j) 
\end{align*}
Which simplifies to
\begin{align}
&\int_{p\inv(j)}\hspace{-1em} f(x)(A\cap p\inv(i))~dp\dg_\mu(j) \nonumber \\
&=p\dg_\nu(i)(A\cap p\inv(i))\int_{p\inv(j)}f(x)(p\inv(i))~dp\dg_\mu(j) \label{eq:simplified}
\end{align}
We can now use \eqref{eq:simplified} to get
\begin{align}
&b\klcirc f\klcirc a(x)(A)\nonumber \\
&=\sum_{i\in X'} \int_{p\inv(\pi(x))} f(x)(A\cap p\inv(i))~dp\dg_\mu(p(x)) \nonumber\\
&=\int_{p\inv(p(x))}f(x)(A)~dp\dg_\mu(p(x))=f\klcirc a(x)(A)
\end{align}
\end{myProof}

\begin{myProof}{Lemma \ref{lem:Levy}}
The map $f(-)(A): X\to \R$ defines a random variable, and the discretization scheme defines a filtration $\sigma(p_n)\subseteq \sigma(p_{n+1})$ whose union is $\mathcal{B}_X$. Following Lemma $\ref{lem:disintegrationCondExp}$ and Proposition \ref{prop:approxEndo} we have
\[
f^n(x)(A):=f(x)(A)\klcirc p_n\dg\klcirc p_n = \EXP{f(x)(A)\mid \sigma(p_n)}
\]
We thus have a sequence $f^n(-)(A)$ of random variables $X\to \R$ which is adapted to the filtration $\sigma(p_n), n\in\N$ by construction. We can now compute for any $m<n$
\begin{align*}
&\EXP{f^n(x)(A)\mid \sigma(p_m)}\\
=&f(x)(A)\klcirc p_n\dg\klcirc p_n\klcirc p_m \dg\klcirc p_m\\
\stackrel{(1)}{=}&f(x)(A)\klcirc p_n\dg\klcirc p_n\klcirc (p_{nm}\klcirc p_n)\dg\klcirc p_m\\
\stackrel{(2)}{=}&f(x)(A)\klcirc p_n\dg\klcirc p_n\klcirc p_n\dg \klcirc p_{nm}\dg\klcirc p_m\\
\stackrel{(3)}{=}&f(x)(A)\klcirc p_m\dg\klcirc p_m=f^m(x)(A)
\end{align*}
where $(1)$ is by definition \eqref{def:discretization}, $(2)$ is by Thm \eqref{thm:KrnDagger} and $(3)$ is by Theorem \eqref{thm:Disintegration}.
We have thus shown that $f^n(-)(A)$ is a martingale for the filtration generated by the discretization scheme, and the result now follows from L\'{e}vy's upward convergence Theorem (\cite[Th. 14.2]{williams1991probability}) since $f(x)(A)=\EXP{f(x)(A)\mid \sigma\left(\bigcup_n \sigma(p_n)\right)}$.
\end{myProof}

\begin{myProof}{Theorem \ref{thm:ApproxConv}}
Let $(B_n)_{n\in\N}$ be a countable basis for the Borel $\sigma$-algebra of $X$, which we assume w.l.o.g. is closed under finite unions and intersections. It follows from Lemma \ref{lem:Levy} that for each $B_n$,  $\lim_k f^{k}(x)(B_n)\sseq f(x)(B_n)$ for all $x\in X\setminus N_n$  where $\mu(N_n)=0$. It follows that for every $x\in X\setminus \bigcup_i N_i$
\[
\lim_{k\to\infty} f^k(x)(B_n)\sseq f(x)(B_n)
\]
for all basic Borel sets $B_n$, and $\mu(\bigcup_i N_i)=0$. Now we use the $\pi-\lambda$-lemma with $(B_n)_{n\in\N}$ as our $\pi$-system. We define
\[
\mathcal{L}:=\{C\mid f^n(x)(C)\to f(x)(C) \text{ for all }x\in X\setminus \cup_i N_i\}
\]
and show that it is a $\lambda$-system. Clearly each $B_n\in \mathcal{L}$. Suppose $C\in \mathcal{L}$, it is then immediate that $C^c\in \mathcal{L}$. Now consider a sequence $C_i\in \mathcal{L}$ with $C_i\subseteq C_{i+1}$, and let $C_\infty:=\cup_{i=1}^\infty C_i$. We want to show that
\begin{align*}
\lim_n f^n(x)(C_\infty)&=\lim_n\lim_m f^n(x)(C_m)\\
&\stackrel{(*)}{=}\lim_m \lim_n f^n(x)(C_m)\\
&=\lim_m f(x)(C_m)\\
&=f(x)(C_\infty)
\end{align*}
where $(*)$ is the only step we need to justify. To show that the iterated limits can be switched, note first that since 
\begin{align*}
&|f^n(x)(C_m)-f(x)(C_\infty)|\\
&=|f^n(x)(C_m)-f(x)(C_m)+f(x)(C_m)-f(x)(C_\infty)|\\
&\leq |f^n(x)(C_m)-f(x)(C_m)|+|f(x)(C_m)-f(x)(C_\infty)|
\end{align*}
since the two terms converge separately, for any $\epsilon>0$ we can find $N>0$ s. th. for all $m,n\geq N$, $|f^n(x)(C_m)-f(x)(C_\infty)|<\nicefrac{\epsilon}{2}+\nicefrac{\epsilon}{2}=\epsilon$. Thus $\lim_{(m,n)\to\infty}f^n(x)(C_m)=f(x)(C_\infty)$. 

Now note also that for all $m_0\in \N$ the sequence $f^n(x)(C_{m_0})$ converges to $f(x)(C_{m_0})$ (by definition of $\mathcal{L}$), and it is not hard to see that $f(x)(C_{m})$ converges to $f(x)(C_\infty)$. Conversely for all $n_0\in\N$, the sequence $f^{n_0}(x)(C_m)$ converges to $f^{n_0}(C_\infty)$ (by virtue of $f^{n_0}(x)$ being a measure). For $\nicefrac{\epsilon}{2}>0$ we can find $N>0$ such that for all $m,n>N$, $|f^n(x)(C_m)-f(x)(C_\infty)|<\nicefrac{\epsilon}{2}$. We can also find $M>0$ such that for all $m>M$, $|f^{n}(x)(C_m)-f^n(x)(C_\infty)|<\nicefrac{\epsilon}{2}$. By taking the maximum of $N$ and $M$ it is clear that for all $m,n$ above this maximum
\begin{align*}
&|f^n(x)(C_\infty)-f(x)(C_\infty)|\\
&\leq |f^n(x)(C_\infty)-f^n(x)(C_m)|+|f^n(x)(C_m)-f(x)(C_\infty)|<\epsilon
\end{align*}
We have thus shown that
\begin{align*}
\lim_{m\to\infty} \lim_{n\to\infty} f^n(x)(C_m)&=\lim_{(m,n)\to\infty}f^n(x)(C_m)\\
&=f(x)(C_\infty)=\lim_{n\to\infty} \lim_{m\to\infty} f^n(x)(C_m).
\end{align*}
Thus $\mathcal{L}$ is a $\lambda$-system, and it follows from the $\pi-\lambda$-lemma that $\sigma((B_n)_{n\in\N})\subseteq \mathcal{L}$ which concludes the proof of pointwise almost everywhere convergence.

For the proof of $\Lps[1]$-convergence we start by showing that
\begin{equation}\label{eq:thm:ApproxConv}
\lim_{n\to \infty}\Norm{\Lpfop f^n(\one_A) - \Lpfop (\one_A)}_1=0
\end{equation}
for any Borel subset $A$. For this we use exactly the same reasoning as above. The only difference is that we need to check that 
\[
\lim_{n\to\infty} \int_X \vert f^n(x)(C_\infty)-f(x)(C_\infty)\vert~d\mu =0
\]
For this we use the fact that we have just shown $f^n(x)(C_\infty)\to f(x)(C_\infty)$ pointwise almost everywhere, and that  $\vert f^n(x)(C_\infty)-f(x)(C_\infty)\vert\leq 1$ with $1$ $\mu$-integrable. It follows by dominated convergence that 
\begin{align*}
&\lim_{n\to\infty} \int_X \vert f^n(x)(C_\infty)-f(x)(C_\infty)\vert~d\mu \\
&=\int_X \lim_{n\to\infty} \vert f^n(x)(C_\infty)-f(x)(C_\infty)\vert~d\mu =0
\end{align*}
which concludes the proof of \eqref{eq:thm:ApproxConv}. To extend the result to simple functions and then to arbitrary functions $\phi\in \Lps[1](X,\nu)$ is routine. 
\end{myProof}

%

\begin{myProof}{Lemma \ref{lem:tensorCont}}
Let $f^n,f: (X_1,\mu_1)\klar (Y_1,\nu_1), g^n,g: (X_2,\mu_2)\klar (Y_2,\nu_2)$, we need to show that for any $\phi \in\Lps[1](Y_1\times Y_2, \nu_1\otimes \nu_2) $
\begin{align}
\int_{X_1\times X_2} \big\vert & \ssint{Y_1\times Y_2}  \phi(y_1,y_2) d(f^n\otimes g^n)(x_1,x_2)- \label{eq:lem:tensorCont:step1}\\
 &\ssint{Y_1\times Y_2}  \phi(y_1,y_2) d(f\otimes g)(x_1,x_2)\big\vert ~d\mu_1\otimes d\mu_2\to 0\nonumber
\end{align}
To show this it is enough to show that 
\[
\ssint{Y_1\times Y_2} \hspace{-3ex} \phi(y_1,y_2)~d(f^n\otimes g^n)(x_1,x_2)\to  \hspace{-3pt} \ssint{Y_1\times Y_2} \hspace{-3ex} \phi(y_1,y_2)~d(f\otimes g)(x_1,x_2)
\] 
pointwise almost everywhere and that
\begin{align}
\sup_n \Norm{\int_{Y_1\times Y_2}  \phi(y_1,y_2) d(f^n\otimes g^n)(x_1,x_2)}_1<\infty\label{eq:lem:tensorCont:step2}
\end{align}
Since in these circumstances pointwise convergence almost everywhere implies $\Lps[1]$-convergence. 

To show \eqref{eq:lem:tensorCont:step1} we proceed as usual: we start with simple functions and use monotone convergence. Let us first consider any measurable $A\subseteq Y_1\times Y_2$, then
\begin{align*}
&\ssint{(y_1,y_2)\in Y_1\times Y_2}\one_A(y_1,y_2)~d(f^n\otimes g^n)(x_1,x_2)\\
=&\ssint{y_1\in Y_1}g^n(x_2)(\{y_2\mid(y_1,y_2)\in A\})~df^n(x_1)\\
\stackrel{(1)}{\to}&\ssint{y_1\in Y_1}g(x_2)(\{y_2\mid(y_1,y_2)\in A\}) ~df^n(x_1)\\
=&\ssint{(y_1,y_2)\in Y_1\times Y_2}\one_A(y_1,y_2)~d(f\otimes g)(x_1,x_2)
\end{align*}
where $(1)$ is by assumption on $f^n,g^n$, Theorem \ref{thm:ApproxConv} and dominated convergence. The result extends completely straightforwardly to all simple functions. Finally for an arbitrary $\phi \in\Lps[1](Y_1\times Y_2, \nu_1\otimes \nu_2) $ we construct a monotone approximating sequence of simple functions $s_i\uparrow \phi$ and use the usual $3\epsilon$ argument to conclude.

To show \eqref{eq:lem:tensorCont:step2} we suppose w.l.o.g. that $\phi(x,y)\geq 0$ and compute using Fubini and Lemma \ref{lem:nonExpansive}
\begin{align*}
&\Norm{\int_{Y_1\times Y_2}  \phi(y_1,y_2)~d(f^n\otimes g^n)(x_1,x_2)}_1\\
=&\ssint{X_1}\ssint{X_2}\ssint{Y_1}\ssint{Y_2}\phi(y_1,y_2)~dg^n(x_2)~df^n(x_1) ~d\mu_2~d\mu_1\\
=&\ssint{X_1}\ssint{Y_1}\ssint{X_2}\ssint{Y_2}\phi(y_1,y_2)~dg^n(x_2)~d\mu_2 ~df^n(x_1)~d\mu_1\\
\leq &\ssint{X_1}\ssint{Y_1}\ssint{X_2}\ssint{Y_2}\phi(y_1,y_2)~dg(x_2)~d\mu_2 ~df^n(x_1)~d\mu_1\\
=&\ssint{X_2}\ssint{Y_2}\ssint{X_1}\ssint{Y_1}\phi(y_1,y_2)~df^n(x_1)~d\mu_1 ~dg(x_2)~d\mu_2\\
\leq & \ssint{X_2}\ssint{Y_2}\ssint{X_1}\ssint{Y_1}\phi(y_1,y_2)~dg(x_2)~df(x_1) ~d\mu_1~d\mu_2\\
=&\Norm{\int_{Y_1\times Y_2}  \phi(y_1,y_2) ~d(f\otimes g)(x_1,x_2)}_1<\infty
\end{align*}
The last step is by definition of $\nu_1,\nu_2$. 
\end{myProof}

\begin{myProof}{Theorem \ref{thm:KleeneStarConv}}
By definition of $(-)^*$ we need to show that $f^n_\infty\SOT f_\infty$ and that post-composing with $\Giry\bigcup$ is continuous for the SOT. In fact, we can shown both by proving that composition is jointly continuous for the SOT  and operators in the image of $\Lpfop$. Note that in general composition is \emph{not} jointly continuous for the SOT, but stochastic operators form a bounded set of operators and on these composition \emph{is} jointly continuous in the SOT. Assume $g^n\SOT g$ for kernels $g^n,g: (X,\mu)\klar (Y,\nu)$ and $h^n\SOT h$ for kernels $h^n,h: (Y,\nu)\klar (Z,\rho)$. We then have for any $\phi\in \Lps[1](Z,\rho)$:
\begin{align*}
&\Norm{\Lpfop[1] g^n(\Lpfop[1] h^n(\phi))-\Lpfop[1] g(\Lpfop[1] h(\phi))}\\
&=\Norm{(\Lpfop[1] g^n-\Lpfop[1] g)(\Lpfop[1] h(\phi))+(\Lpfop[1] g^n(\Lpfop[1] h^n(\phi)-\Lpfop[1] h(\phi)))}\\
&\leq \Norm{(\Lpfop[1] g^n-\Lpfop[1] g)(\Lpfop[1] h(\phi))}+\Norm{(\Lpfop[1] g^n(\Lpfop[1] h^n(\phi)-\Lpfop[1] h(\phi)))}\\
&\leq \Norm{\Lpfop[1] g^n-\Lpfop[1] g)(\Lpfop[1](h(\phi))}+\Norm{\Lpfop[1] g^n}\Norm{\Lpfop[1] h^n(\phi)-\Lpfop[1] h(\phi)} 
\end{align*}
where the last step follows from H\"{o}lder's inequality and the fact that $\Lpfop[1] g^n$, being a stochastic operator, is bounded.

Continuity under post-composition by $\Giry\bigcup$ now follows easily. For the continuity of the $(-)_\infty$ construction, we first work inductively on the construction of the maps $a_k$ defined in \eqref{eq:inftyDef1} and \eqref{eq:inftyDef2}. We denote by $a_k^n$ the kernels generated by $f^n$ and $a_k$ those generated by $f$, $k\in \N$. By definition \eqref{eq:inftyDef1} $a_1^n=\eta\otimes f^n\klcirc \Delta_1$ and it follows from Lemma \ref{lem:tensorCont} and SOT continuity of composition that $a_1^n\SOT a_1$. Now assuming that $a_{k-1}^n\SOT a_{k-1}$, the exact same argument show that $a_{k}^n\SOT a_{k}$. Finally, the joint continuity of composition gives
\begin{align}\label{eq:thm:KleeneStarConv}
b_k^n:=a_k^n\klcirc \ldots\klcirc a_1^n\SOT b_k:=a_k\klcirc\ldots\klcirc a_1
\end{align}
for every $k\in \N$.

Having shown that the maps $b_k^n$ defining $f^n_\infty$ converge to the maps $b_k$ defining $f_\infty$ we now show that $\Norm{\Lpfop f^n_\infty(\phi)-\Lpfop f_\infty(\phi)}_1$ converges to 0 for any $\phi\in \Lps[1]((2^H)^\omega,\mu^\omega)$. As usual we start with simple functions: let $\phi=\sum_i \alpha_i \one_{B_i}$ be a simple function. For any $B_i$ we have
\[
B_i\subseteq \prod_{k}^\infty \pi_k[B_i]\Rightarrow\mu^\infty(B_i)\leq \prod_k^\infty \mu(\pi_k)
\]
and thus there can only exist finitely many indices for which the corresponding projection of $B_i$ has not got full measure. Since the simple function is a finite sum, this means that $\phi$ is defined by measurable sets with non-trivial measure on only a finite set of coordinates. Let $N$ be the largest of these coordinates, we then have:
\begin{align*}
&\Norm{\Lpfop f^n_\infty (\phi)-Lpfop f_\infty (\phi)}_1\\
=&\ssint{2^H} \big\vert \ssint{(2^H)^\omega} \phi~ df^n_\infty(x) - \ssint{(2^H)^\omega} \phi ~df_\infty(x) \big\vert d\mu\\
=&\ssint{2^H} \big\vert \ssint{(2^H)^N} \phi ~d b_N^n(x) - \ssint{(2^H)^N} \phi ~db(x) \big\vert ~d\mu
\end{align*}
and the result follows from \eqref{eq:thm:KleeneStarConv}. Extending to arbitrary maps is completely straightforward.
\end{myProof}

\end{document}